\documentclass{theoretics}

\title{On small-depth Frege proofs for \texorpdfstring{$\PHP$}{PHP}}

\ThCSauthor{Johan H{\aa}stad}{johanh@kth.se}[0000-0002-5379-345X]
\ThCSaffil{KTH Royal Institute of Technology, Stockholm, Sweden}
\ThCSthanks{A preliminary version of this article appeared at FOCS 2023 \cite{jhphpfocs}.\\ Work supported by the Knut and Alice Wallenberg Foundation.}
\ThCSshortnames{J.\ H{\aa}stad}  
\ThCSshorttitle{On small-depth Frege proofs for $\PHP$}
\ThCSyear{2025}
\ThCSarticlenum{27}
\ThCSreceived{Jan 31, 2024}
\ThCSrevised{Jul 2, 2025}
\ThCSaccepted{Oct 14, 2025}
\ThCSpublished{Dec 1, 2025}
\ThCSdoicreatedtrue
\ThCSkeywords{Pigeon Hole Principle, small-depth formulas, Frege proofs, switching lemma.}

\usepackage{url}
\usepackage{graphicx}
\usepackage{textcomp}
\usepackage{amsfonts}
\usepackage{float}
\usepackage{tikz}
\usepackage{thm-restate}
\DeclareMathOperator{\PHP}{PHP} 

\makeatletter
\newcommand\IfRestateTF{%
  \ifx\label\thmt@gobble@label 
    \expandafter\@firstoftwo
  \else
    \expandafter\@secondoftwo
  \fi
}
\makeatother
\newcommand{\RestateRemark}{\IfRestateTF{{\normalfont\bfseries (Restated) }}{}}

\begin{document} 
\maketitle

\begin{abstract}
We study Frege proofs for the functional and onto
graph Pigeon Hole Principle
defined on the $n\times n$ grid where $n$ is odd.
We are interested in the case where each formula
in the proof is a depth $d$ formula in the basis given by
$\land$, $\lor$, and $\neg$.  We prove that in this situation the
proof needs to be of size exponential in $n^{\Omega (1/d)}$.
If we restrict each line in the proof to be of
size $M$ then the number of lines needed is exponential
in $n/(\log M)^{O(d)}$.   The main technical component of
the proofs is to design a new family of random restrictions
and to prove the appropriate switching lemmas.
\end{abstract}

\section{Introduction}
In this paper we study formal proofs of formulas in
Boolean variables encoding natural combinatorial principles.
We can think of these as tautologies but it is often
more convenient to think of them as contradictions.
When a certain formula, $F$, is a contradiction then its negation,
$\bar F$ is a tautology and we do not distinguish the two.
In particular in the below discussion we might call something
a tautology that the reader, possibly rightly, thinks of
as a contradiction.  We are equally liberal with usage
of the word ``proof'' which might more accurately be
called ``a derivation of contradiction''.  

We are given a set of local constraints that we call
``axioms''.  These are locally satisfiable but not globally in
that there is no global assignment that satisfies all the axioms.
A proof derives consequences
of the axioms and it is complete when it reaches 
an obvious contradiction such as $1=0$ or that an
empty clause contains a true literal.  

A key property of such a proof system is the kind
of statements that can be used and in
this paper we study Boolean formulas with at most $d$ alternations.
Traditionally one has studied formulas over the basis $\land$,
$\lor$, and $\neg$, where negations only appear at the 
inputs.  In this formalism one counts the number of alternations
between $\land$ and $\lor$.  We instead use the basis given
by $\lor$ and $\neg$ and count the number of alternations between
the two connectives.  There is an easy translation between the
two as $\land$ can be simulated by $\neg \lor \neg$.  We 
consider the case where $d$ is a constant or a slowly growing
function of the size of the input.

A fundamental and popular case is resolution
corresponding to $d=1$, where each formula is a disjunction of literals. 
It is far from easy to analyze resolution
but this proof system has been studied for a long time
and many questions are now resolved.  We do not want to
discuss the history of resolution but as it is very relevant
for the current paper let us mention that an  early milestone
was obtained by Haken \cite{haken} in 1985 when he proved
that the Pigeon Hole Principle ($\PHP$) requires exponential
size resolution proofs.  The $\PHP$ states that $n+1$ pigeons can 
fly to $n$ holes such that no two pigeons fly to the
same hole.  There are many ways to code this statement
using Boolean variables and we use the basic variant
which has $(n+1)n$ Boolean variables $x_{ij}$.  
Such a variable is true 
iff pigeon $i$ flies to the hole $j$.
The axioms say that for each $i$ there is a value
of $j$ such that $x_{ij}$ is true and for 
each $j$ there is at most one $i$ such that $x_{ij}$
is true.  This is clearly a contradiction but to
prove this counting is useful and resolution is not very 
efficient when it comes to counting.

\paragraph{Previous work.} The focus of this paper is the more
powerful proof system obtained for larger values of 
$d$ and here a pioneering result was obtained by 
Ajtai \cite{Ajtai94Complexity} proving superpolynomial lower bounds 
for the size of any proof for $\PHP$ for any fixed constant $d$.
The lower bounds of Ajtai were not explicit and Bellantoni
et al \cite{BPU92ApproximationAS}
gave the first such bounds, namely that depth $\Omega(\log^* n)$ is
needed for the size of the proof to be polynomial.
This was greatly improved 
in two independent works
by Kraj\'{i}\v{c}ek, Pudl\'{a}k, and Woods \cite{KPW95FregePHPExp}
and Pitassi, Beame, and Impagliazzo \cite{PBI93ExponentialLowerBounds},
respectively.  These two papers established lower bounds
for the size of any proof of the $\PHP$
of the form exponential in $n^{c^{-d}}$ where $c>1$, and gave
non-trivial bounds for depths as high as $\Theta (\log {\log n})$.

Related questions were studied in circuit complexity where the
central question is to study the size of a circuit needed
to compute a particular function.  Here a sequence of
results \cite{fss},\cite{sipser}, \cite{yao}, \cite{jhswitch}
established size lower bounds of
the form exponential in $n^{\Omega (1/d)}$ and obtained
strong lower bounds for $d$
as large as $\Theta (\log n / \log {\log n})$.  
The results were obtained for the parity function and
here it is easy to show that this function can be computed
by circuits of matching size.  To see that $\PHP$ 
allows proof of size exponential in $n^{O(1/d)}$ is
more difficult but was established in 2001 by
Atserias et al \cite{agg01}.

A technique used in many of these papers is called ``restrictions''.
The idea is simply to, in a more or less
clever way, give values to most of the
variables in the object under study and to analyze
the effect.  One must preserve\footnote{One does not really
preserve a function or a formula and an object of size $n$ is
reduced to a similar object of size $f(n)$ for some $f(n) < n$.}
the function computed (or tautology being proved) while 
at the same time be able to
simplify the circuits assumed to compute the function or
the formulas in the claimed proof.
An important reason that the lower bounds in circuit complexity
were stronger than those in proof complexity
is that it is easier to preserve a single function
than an entire tautology with many axioms.
Due to these complications strong lower bounds for
depths beyond $\Theta( \log {\log n})$, for any tautology, remained unknown for 
several decades.  

The first result that broke this barrier was 
obtained by Pitassi et al \cite{pitassi16frege}
who obtained super-polynomial lower bounds for
depths up to any $o(\sqrt {\log n})$.
The tautology investigated was first studied by Tseitin \cite{tseitin}
and considers a set of linear equations modulo two defined
by a graph.  The underlying graph for
\cite{pitassi16frege} is an expander.  These results were
later extended to depth $\Theta( \log n / \log {\log n})$
by H{\aa}stad \cite{jhtseitin}
and in this case the underlying graph is the two-dimensional square grid.
The bounds obtained were further improved by 
H{\aa}stad and Risse \cite{jhkr}.

All results mentioned so far only discuss total size.  For
resolution, each formula derived is a clause and hence of size 
at most $n$ but
for other proof systems it is interesting to study
the number of lines in the proof and the sizes of lines
separately.  Pitassi, Ramakrishnan,
and Tan \cite{PRT21} had the great insight that a technical
strengthening of the used methods yields much
stronger bounds for this measure than implied by the
size bounds.  They combined some of the techniques
of \cite{jhtseitin} with methods from \cite{pitassi16frege}
to establish that if each line is of size at most $M$ then
the number of lines in a proof that establishes the Tseitin
principle over the square grid needs to be exponential in $n2^{-d \sqrt {\log M}}$.
By using some additional ideas H{\aa}stad and Risse \cite{jhkr}
fully extended the techniques of \cite{jhtseitin}
to this setting improving the bounds to exponential
in $n/(\log M)^{O(d)}$.

\paragraph{Our results.}
Despite this progress, the over thirty year old question whether the
$\PHP$ allows polynomial size proofs of depth $O( \log \log n)$
remained open.  The purpose of this paper is to prove
that it does not and that lower bounds similar to
those for the Tseitin tautology also apply to the
$\PHP$.   To build on previous techniques we study what is known as the
graph $\PHP$ where the underlying graph is an odd size
two-dimensional grid.  

As the side length of the grid is odd, if one colors it
as a chess board,  the corners are of the same color and
let us assume this is white.
In the graph $\PHP$ on
the grid, there is a pigeon on each white square and
it should fly to one of the adjacent black squares
that define the holes.   This graph $\PHP$ is the result
of the general $\PHP$ where most variables are forced to 
take the value $0$.
Each pigeon is only given at most four alternatives.
Clearly any proof for the general $\PHP$ can be modified
to give a proof of the graph $\PHP$ by replacing some variables
by the constant false.  To limit ourselves
further we prove our lower bounds for what is known
as the functional and onto $\PHP$.  We add axioms saying
that each pigeon can only fly to one hole and each hole
receives exactly one pigeon.

Phrased slightly differently, the functional and onto
$\PHP$ on the grid says that
there is a perfect matching of the odd size grid and we heavily use
local matchings.   We can compare this to
the Tseitin tautology on the grid studied
by \cite{jhtseitin, PRT21, jhkr} that states that it is possible to
assign Boolean values
to the edges of the grid such that there is an odd number
of true variables next to any node.  As a perfect matching
would immediately yield such an assignment, the $\PHP$ is a stronger
statement and possibly easier to refute.
In particular, any lower bound for functional and onto $\PHP$ gives
the same lower bound for the Tseitin contradiction.
We do not, however, think of the results of the current paper
as a new proof for the results in  \cite{jhkr} (with slightly weaker bounds).
In fact, as the construction
of this paper shares many features with the construction
of \cite{jhkr}, it is better to think the current proof
as an adjustment of the latter proof taking care of complications
that we here only work with matchings.
Let us turn to discuss the main technical point,
namely to prove a ``switching lemma''.

By assigning values to most variables in a formula it is
possible to switch a small depth-two formula from being
a CNF to being a DNF and the other way around.
In the basic switching lemma used to prove circuit lower bounds
\cite{fss,yao, jhswitch}, uniformly random constant values 
are substituted
for a majority of the variables.  Such restrictions
are the easiest to analyze but are less useful in 
proof complexity as they do not preserve any interesting tautology.

To preserve a tautology or a complicated
function it is useful to replace several old
variables by the same new variable, possibly negated.
It is possible to be even more liberal and allow 
old variables to be replaced by slightly more complicated
expressions in the new variables.  This technique
was first introduced explicitly by Rossman, Servedio,
and Tan \cite{rst} when studying the depth hierarchy
for small-depth circuits but had been used
in more primitive form in earlier papers.
We use such generalized restrictions in this paper.

\paragraph{Multi-switching.}
The technical strengthening needed by \cite{PRT21} that
we discussed above is to improve the standard switching
lemma to what is commonly
known as a multi-switching lemma.  This concept was first introduced 
independently by H{\aa}stad \cite{jhmultiswitch}
and Impagliazzo, Mathews, and Paturi \cite{imp} to
study the correlation of small-depths circuits and 
simple functions such as parity.

In this setting one considers many formulas $(F^i)_{i=1}^m$
and the goal is to switch them all simultaneously
in the following sense.  There is a small depth (common)
decision tree such that at any leaf of the tree it
is possible to represent each $F^i$ by a small
formula of the other type.
It was the insight that multi-switching
can be used in the proof complexity setting that
made it possible for \cite{PRT21} to derive the
strong bounds on the number of lines in a proof
when each line is short.   The extension to
multi-switching turns out to be mostly technical
and not require any really new ideas.  In view
of this let us for the rest of this introduction
only discuss standard switching.

\paragraph{Techniques used.}
The most novel part of this paper is to design a new space
of restrictions that preserve the functional and
onto graph $\PHP$ on the grid.  It has many similarities
with the space introduced in \cite{jhtseitin}
and from a very high level point of view, the
proofs follow the same path.  At the more detailed
level in this paper we work with partial
matchings of the grid which is a more rigid
object than assignments that only satisfy the 
Tseitin condition of an odd number of true
variables next to any node.  This results
in considerable changes in the details
and as a result we get slightly worse bounds.

Once the space of restrictions is in place, two
tasks remain.  Namely, to prove the switching lemmas
and then use these bounds to derive the claimed
bounds on proof size.  This latter part hardly
changes compared to previous papers.  The switching lemmas follow
the same pattern as in \cite{jhtseitin} and \cite{jhkr}
but we need some new combinatorial lemmas and the
fact that we work with matchings calls
for some modifications.

\paragraph{Outline of the paper.}
We start with some preliminaries and recall some facts
from previous papers in Section~\ref{sec:prelim}.
We introduce our new space of random
restrictions in Section~\ref{sec:restrictions}.
The basic switching lemma is proved in
Section~\ref{sec:switch} and we use
it to establish the lower bound for proof size
in Section~\ref{sec:size}.  We give
the multi-switch lemma in Section~\ref{sec:mswitch} and
use it, in Section~\ref{sec:nlines}, to derive the lower bounds on the
number of lines in a proof.
We end with some very brief comments in Section~\ref{sec:conclusion}.
This is the final version of \cite{jhphpfocs} which was presented
at the 2023 FOCS conference.

\section{Preliminaries} \label{sec:prelim}

In this section we give some basic definitions and derive some
simple properties.  We also recall some
useful facts from related papers.

\subsection{The formula to refute}

We study the functional and onto $\PHP$ 
on the odd size $n\times n$ grid and when we want to emphasize
the size of the grid we use the term $\PHP_n$. Nodes are given indexed
by $(i,j)$ where $1 \leq i,j \leq n$ and a node is connected
to other nodes where one of the two coordinates is the same
and the other differ by $1$.  As opposed to some previous
papers studying the Tseitin tautology,
it is here important that we are on the grid and not on
the torus as we want a bipartite graph.  For any edge, $e=(v,w)$ of adjacent
nodes, $v$ and $w$, we have a variable $x_{e}$.  The axioms say that for each
$v$ if you consider the, at most four, variables adjacent to $v$
exactly one of them is true.  This is easy to write as
a 4-CNF.  The axioms look exactly the same
for holes and pigeons.  

This formula is a stronger statement compared to
the standard $\PHP$ formula with
$(n^2+1)/2$ pigeons and $(n^2-1)/2$ holes.  Our
parameters are slightly different from the standard parameters
where $n$ is the number of holes, but we trust the reader to
keep this in mind.

We assume that the grid is colored as a chess board and
that the corners are white.  Thus, pigeons correspond to white
nodes and holes to black nodes.  We prefer to use the colors to
distinguish the two as we later use larger objects that inherit
the same colors.

\subsection{Frege proofs}\label{sec:frege}

We consider proofs where each line in the proof
is either an axiom or derived from previous lines.
The derivation rules are not important and all we need
is that they are of constant size and sound.  
We use the same rules as \cite{pitassi16frege}, \cite{jhtseitin},
\cite{PRT21}, and \cite{jhkr}.   We demand that each formula that appears
is of depth at most $d$ and, as several previous papers, we do not allow
$\land$ and count the number of alternations of $\neg$ and $\lor$.
The $\land$ operator is simulated by $\neg \lor \neg$.
The rules are as follows.

\begin{itemize}

\item  (Excluded middle) $(p \lor \neg p)$

\item  (Expansion rule) $(p \rightarrow p \lor q)$

\item (Contraction rule) $(p \lor p) \rightarrow p$

\item (Association rule) $p \lor (q \lor r) \rightarrow (p \lor q) \lor r$

\item (Cut rule) $p \lor q, \neg p \lor r \rightarrow q \lor r$. 

\end{itemize}

\subsection{Decision trees and partial matchings of the grid}

As the use of decision trees is central to the current
paper let us define some notions.
In a standard decision tree, at each node we ask for 
the value of an input variable.  In this paper we instead allow
queries of the form.
\begin{center}
``To which node is $v$ matched?''
\end{center}

The answer to this question determines the value
of any variable next to $v$ and thus is more powerful
than a single variable question.  On the other hand 
it can be simulated by asking
ordinary variable questions for three variables
around $v$.  Thus within a factor of three in the
number of questions, this type of questions is equivalent
to variable queries.

We require that the answers along any branch in a decision
are locally consistent and keep all such branches.
The key property is that it is possible to extend
the partial matching given by the answers on this branch
to a large fraction of the grid. 
Product sets of collections of disjoint sets
of intervals are convenient for us.
The size of an interval is the number of elements
it contains.  The size of a partial matching is the
number of edges it contains.

\begin{definition}\label{def:consistent}
A partial matching $M$ on the $n \times n$ grid of
size $t$ is {\em locally consistent} 
if it can be extended to a complete
matching of a larger set $S \times T$.
We require that each of $S$ and $T$ is the union
of disjoint even-sized intervals such that the total
size of all intervals in each of $S$ and $T$ is at most $48t$.
\end{definition}

\paragraph{Remark.}  The constant $48$ in the definition is not optimal
and with a more careful analysis it can be decreased, but we
have opted for shorter proofs rather than optimal constants.

\medskip
We use two properties from locally consistent
matchings.  Firstly that given any node not matched
in the matching then it is possible to include
also this node with a suitable partner.  This is proved
in Lemma~\ref{lemma:cons} below.   The second property 
we need is that a sub-matching of any locally consistent
matching is also locally consistent.  This is proved
in Lemma~\ref{lemma:conssub}.  We do not exclude
that there is a simpler definition of locally consistent 
that achieve these two properties, but our
definition is one possibility.

Most of the time we think of parts of the grid as sets of
points in $\mathbb R^2$ with integer coordinates connected
in the natural way.  For pictures it is sometimes convenient
to view them as black and white unit size squares connected
if they share and edge.  We sometimes use this viewpoint.

\begin{lemma}\label{lemma:cons}
Suppose we have a locally consistent partial matching, $M$ of size at
most $n/50-9$ in the $n \times n$ grid and that we are given
a node $v$ not matched
by $M$.  It is then possible to find a partner, $w$, of $v$, 
such that $M$ jointly with $(u,w)$ is a locally
consistent matching.
\end{lemma}

\begin{proof}
If $v$ is already in $S \times T$ we can use the same extension.
Suppose $v=(a,b)$ where $a \in S$ and $b \not\in T$.
It is easy to find $b'$ such that $T \cup \{ b, b'\}$
is a union of even size intervals.  Now we can add matchings
of $S \times b$ and $S \times b'$ using that $S$ is a union
of even size intervals.

The case when $a \not\in S$ and $b \in T$ is symmetric
and let us handle the case $a \not \in S$ and $b \not\in T$.
We can find $b'$ as in the previous case, enlarging $T$
to $T'= T \cup \{ b, b'\}$ and then proceed by
adding $a$ and a suitable $a'$ to $S$.

We enlarge the sizes of each of $S$ and $T$ by at most 2 but since
the size of the matching increases by one this is not a problem.
\end{proof}

Let us state the natural property that
any sub-matching of a locally consistent matching is
locally consistent.  As our proof is surprisingly
complicated we postpone the proof to Appendix~\ref{sec:proof}.

\begin{restatable}{lemma}{Conslemma}\label{lemma:conssub}\RestateRemark
Suppose $M$ is a locally consistent matching.
Then any subset of $M$ is locally consistent.
\end{restatable}

We are also interested in matching areas of the plane
which are almost a complete square.  The difference
is that some squares next to the perimeter are
removed.  We call such a removed square a ``dent''.

\medskip
\begin{figure}[h]
\begin{center}
\begin{tikzpicture}[scale=1.288]
\draw (0,0) rectangle(0.5,.5);
\fill[{lightgray}] (0,0) rectangle(0.5,.5);
\fill[{lightgray}] (1.0,0) rectangle(1.5,.5);
\fill[{lightgray}] (0.5,0.5) rectangle(1.0,1.0);
\fill[{lightgray}] (1.5,0.5) rectangle(2.0,1.0);
\fill[{lightgray}] (2.5,0.5) rectangle(3.0,1.0);
\fill[{lightgray}] (3.5,0.5) rectangle(4.0,1.0);
\fill[{lightgray}] (0.0,1.0) rectangle(0.5,1.5);
\fill[{lightgray}] (1.0,1.0) rectangle(1.5,1.5);
\fill[{lightgray}] (2.0,1.0) rectangle(2.5,1.5);
\fill[{lightgray}] (3.0,1.0) rectangle(3.5,1.5);
\fill[{lightgray}] (0.5,1.5) rectangle(1.0,2.0);
\fill[{lightgray}] (1.5,1.5) rectangle(2.0,2.0);
\fill[{lightgray}] (2.5,1.5) rectangle(3.0,2.0);
\fill[{lightgray}] (0.0,2.0) rectangle(0.5,2.5);
\fill[{lightgray}] (1.0,2.0) rectangle(1.5,2.5);
\fill[{lightgray}] (2.0,2.0) rectangle(2.5,2.5);
\fill[{lightgray}] (3.0,2.0) rectangle(3.5,2.5);
\fill[{lightgray}] (0.5,2.5) rectangle(1.0,3.0);
\fill[{lightgray}] (1.5,2.5) rectangle(2.0,3.0);
\fill[{lightgray}] (2.5,2.5) rectangle(3.0,3.0);
\fill[{lightgray}] (3.5,2.5) rectangle(4.0,3.0);
\fill[{lightgray}] (0.0,3.0) rectangle(0.5,3.5);
\fill[{lightgray}] (1.0,3.0) rectangle(1.5,3.5);
\fill[{lightgray}] (2.0,3.0) rectangle(2.5,3.5);
\fill[{lightgray}] (3.0,3.0) rectangle(3.5,3.5);
\fill[{lightgray}] (0.5,3.5) rectangle(1.0,4.0);
\fill[{lightgray}] (1.5,3.5) rectangle(2.0,4.0);
\fill[{lightgray}] (2.5,3.5) rectangle(3.0,4.0);
\fill[{lightgray}] (3.5,3.5) rectangle(4.0,4.0);

\fill[black] (0.5,0) rectangle(1,.5);
\draw (.5,.5) rectangle(1,1);

\draw (1,0) rectangle(1.5,.5);
\fill[black] (1.5,0) rectangle(2,.5);
\fill[black] (1,.5) rectangle(1.5,1);
\draw (1.5,.5) rectangle(2,1);

\fill[black] (2.5,0) rectangle(3,.5);
\fill[black] (2,.5) rectangle(2.5,1);
\draw (2.5,.5) rectangle(3,1);

\fill[black] (3.5,0) rectangle(4,.5);
\fill[black] (3,.5) rectangle(3.5,1);
\draw (3.5,.5) rectangle(4,1);

\draw (0,1) rectangle(0.5,1.5);
\fill[black] (0.5,1) rectangle(1,1.5);
\fill[black] (0,1.5) rectangle(0.5,2);
\draw (.5,1.5) rectangle(1,2);

\draw (1,1) rectangle(1.5,1.5);
\fill[black] (1.5,1) rectangle(2,1.5);
\fill[black] (1,1.5) rectangle(1.5,2);
\draw (1.5,1.5) rectangle(2,2);

\draw (2,1) rectangle(2.5,1.5);
\fill[black] (2.5,1) rectangle(3,1.5);
\fill[black] (2,1.5) rectangle(2.5,2);
\draw (2.5,1.5) rectangle(3,2);

\draw (3,1) rectangle(3.5,1.5);
\fill[black] (3.5,1) rectangle(4,1.5);
\fill[black] (3,1.5) rectangle(3.5,2);

\draw (0,2) rectangle(0.5,2.5);
\fill[black] (0.5,2) rectangle(1,2.5);
\fill[black] (0,2.5) rectangle(0.5,3);
\draw (.5,2.5) rectangle(1,3);

\draw (1,2) rectangle(1.5,2.5);
\fill[black] (1.5,2) rectangle(2,2.5);
\fill[black] (1,2.5) rectangle(1.5,3);
\draw (1.5,2.5) rectangle(2,3);

\draw (2,2) rectangle(2.5,2.5);
\fill[black] (2.5,2) rectangle(3,2.5);
\fill[black] (2,2.5) rectangle(2.5,3);
\draw (2.5,2.5) rectangle(3,3);

\draw (3,2) rectangle(3.5,2.5);
\fill[black] (3.5,2) rectangle(4,2.5);
\fill[black] (3,2.5) rectangle(3.5,3);
\draw (3.5,2.5) rectangle(4,3);

\draw (0,3) rectangle(0.5,3.5);
\fill[black] (0.5,3) rectangle(1,3.5);
\fill[black] (0,3.5) rectangle(0.5,4);
\draw (.5,3.5) rectangle(1,4);

\draw (1,3) rectangle(1.5,3.5);
\fill[black] (1.5,3) rectangle(2,3.5);
\draw (1.5,3.5) rectangle(2,4);

\draw (2,3) rectangle(2.5,3.5);
\fill[black] (2.5,3) rectangle(3,3.5);
\draw (2.5,3.5) rectangle(3,4);

\draw (3,3) rectangle(3.5,3.5);
\fill[black] (3.5,3) rectangle(4,3.5);
\fill[black] (3,3.5) rectangle(3.5,4);
\draw (3.5,3.5) rectangle(4,4);
\end{tikzpicture}
\end{center}
\caption{A square with 3 white and 3 black dents.}
\end{figure}

\begin{lemma} \label{lemma:matchcenter}
Suppose we have a square with even side length 
and which has the same number of
white and black dents.
Suppose further that there are at most $M$ dents and no dent is within $M$
of a corner.  In this situation there is always a matching of the
square.
If we have a square with odd side length and white corners,
then the similar statement is true assuming we have 
one more white dent.
\end{lemma}

\begin{remark}
The condition of any dent being at distance at least
$M$ from the perimeter is not exactly sharp.
Some condition on the number of dents next to
a corner is, however, needed as it is easy to see that removing many white
squares next to a corner gives a local surplus of black
squares that cannot be matched.  As the exact condition
is not important we opted for a simple one.
\end{remark}

\begin{proof}
We could use the general condition for the existence 
of a matching given in Lemma~\ref{lemma:nomatch} but
as the situation is very structured we find a direct
proof to be easier to follow.  It also gives better
constants (even though, of course, this is not very important).

Let us first do the case of even side length and
start by sketching the argument.  We construct
a matching in an onion-like fashion using as many pairs
along the perimeter as possible.  Let us say that two dents
are ``adjacent'' if the perimeter is dent-free between these
two nodes.  If two adjacent dents are 
at an even distance then it is possible to perfectly
match all nodes between the two dents.  This happens
iff the two adjacent dents are of different colors.
If two adjacent dents are of the same color we need
to match one node with the node towards the interior
creating a new dent in the new square which has side
length two less than the original square.  Let us
say this more formally.

Suppose the side length is $S$.
Take any dent, and assume for concreteness that it is white.
Start matching nodes two and two
along the perimeter starting with the node next to this dent. 
This is straightforward until we hit the next dent.
If this dent is black we get a perfect match along the perimeter
while if it is white we are forced to create a new white dent.
We continue on the other side of this dent and go all around
the square.  We get a new square of side length $S-2$ and
some dents.  A dent remains iff it is the same color 
as the dent preceding it.

As we have both black and white dents, the number
of dents has decreased by at least two and the distance
to the corner has decreased by at most 2.  As each
matched pair contains one black square and one white
square the number of white dents equals the number
of black dents.
We repeat this process until there are no remaining dents.
The rest is easy to match.

If the side length is odd then the process stops
when there is only one white dent on the boundary.
Also in this case the remaining square is easy to
match.
\end{proof}
It is easy to see that in most situation where we
apply Lemma~\ref{lemma:matchcenter} there are many possible
matchings.  It is convenient for us
to think of the matching as given uniquely by the boundary conditions.
This can be done in many ways and any deterministic procedure
to go from the set of dents to the matching works
for us.

Lemma~\ref{lemma:cons} ensures that in any decision tree with questions of
the form ``To which node is $v$ matched?'' there is one
possible answer such that the partial matching created
by the path so far results in a locally consistent matching.
By simply erasing any branch which is not locally consistent
we maintain the property of each branch being
locally consistent.   

Lemma~\ref{lemma:conssub} ensures a slightly more subtle property.
Suppose we are given a decision tree, $T$,  of small depth and we have
a small partial matching, $\tau$ and we want to explore $T$ given
$\tau$.   Lemma~\ref{lemma:cons} makes sure that as long
as the sum of the size of $\tau$ and the depth of the decision
tree is small, it is always possible to find a locally
consistent partner for any node
queried by $T$.  Lemma~\ref{lemma:conssub} makes sure that 
any such answer is also present in $T$ and hence it has not
already been pruned away.  Let us state this as a lemma

\begin{lemma}\label{lemma:extend}
Let $T$ be a decision tree of depth $t$ and $\tau$
a locally consistent matching of size $t'$.
Then, provided $t+t' \leq n/50$ the
decision tree obtained from pruning all branches
of $T$ which are not consistent with $\tau$
is a non-empty decision tree.
\end{lemma}

We use this lemma in many places without explicitly referring
to it.  We denote the resulting decision tree by $T \lceil_{\tau}$.

We say that a decision tree is a $1$-tree if all its
leafs are labeled one and similarly we have $0$-trees.
It might seem redundant to allow such trees but 
when doing operations on decision trees they naturally occur.

\subsection{\protect{\texorpdfstring{$t$}{t}-evaluations}}

The concept of $t$-evaluations was introduced by Krajíček et al.\
\cite{KPW95FregePHPExp} and is a very convenient tool for proving
lower bounds on proof size. Here we follow the presentation of Urquhart and Fu
\cite{uf} while using the notation of 
\cite{jhtseitin} and \cite{jhkr}.
A $t$-evaluation, $\varphi$, is a map from a set of formulas to decision 
trees of depth at most $t$.  We always have the property that this
set of formulas is closed under taking sub-formulas.
We want the following properties.

\begin{enumerate}

\item The constant true is represented by a $1$-tree and the
constant false is represented by a $1$-tree.

\item If $F$ is an axiom of the $\PHP$ \label{axiomone}
contradiction then $\varphi(F)$ is a $1$-tree.

\item If $F$ is a single variable then
$\varphi(F)$ is the natural decision tree of depth $1$ defining
the value of this variable.

\item If $\varphi(F)=T$ then $\varphi(\neg F)$
is a decision tree with the same internal structure as $T$
but where the value at each leaf is negated.

\item Suppose $F= \lor F_i$. \label{rule:or}
Consider a leaf in $\varphi (F)$ and the assignment, $\tau$ leading to this
leaf.  If the leaf is labeled $0$ then for each $i$ 
$\varphi(F_i) \lceil_\tau$ is a $0$-tree and if the leaf is
labeled $1$ then for some $i$, $\varphi(F_i) \lceil_\tau$ is a $1$-tree.

\end{enumerate}

Remember that in our situation the nodes of the decision tree 
ask to which other node a certain node $v$ is matched and we
also require the assignment along a path in the decision tree
to be locally consistent. In this situation,  
condition \ref{axiomone} follows from
the other conditions, but in general this is not the case.
The key property for $t$-evaluations is the following
lemma.

\begin{lemma}\label{lemma:noproof}
Suppose we have a derivation using the rules
of Section~\ref{sec:frege} starting with the
axioms of the functional onto $\PHP$ on the $n \times n$ grid.
Let $\Gamma$ be the set of all sub-formulas of this derivation
and suppose  there is a $t$-evaluation whose range includes $\Gamma$ where
$ t \leq n/150$. 
Then each line in the derivation is mapped to a $1$-tree.
In particular we do not reach a contradiction.
\end{lemma}

\begin{proof}
This is the lemma that corresponds to Lemma 6.4 of
\cite{jhtseitin}.
It relies on two properties, namely that each axiom is represented
by a $1$-tree, and that the derivation rules preserve this property.
The first property is ensured by the local consistency
of any branch in a decision tree.  The second property
follows from the fact that the derivation rules are sound
and we never ``get stuck'' in a decision tree.  By this
we mean that it always possible to continue a branch
in a decision tree keeping the values locally consistent.
This is ensured by Lemma~\ref{lemma:cons} and Lemma~\ref{lemma:conssub}.

The proof is by induction over the number of lines in the proof
and we need to discuss each of the five rules.  As this is
rather tedious we here only discuss the Cut rule which
is the most interesting rule.  We are confident
that the reader can handle the other rules and if this is not
the case, the other cases are discussed in detail in the
proof of Lemma 6.4 of \cite{jhtseitin}.

The application of the cut rule derives $F= q \lor r$ from
$p \lor q$ and $\neg p \lor r$.  By induction $\varphi(p \lor q)$
and $\varphi (\neg p \lor r)$ are both $1$-trees and we must
establish that so is $\varphi (F)$.
For contradiction, take a supposed leaf with label $0$ 
in $\varphi(F)$ and let $\tau$ be the assignment leading to this leaf. 
We know that $\varphi(q)\lceil_\tau$ and $\varphi(r)\lceil_\tau$ are
both $0$-trees.
Consider any path in $\varphi (p) \lceil_\tau$ and let $\tau_1$ be
the assignment of this path.  Assume this leaf is
labeled $0$, the other case being similar.  Now take
any path in $\varphi(p \lor q)\lceil_{\tau \tau_1}$.
As this is a $1$-tree the label at this path must
be $1$.  This contradicts that $\varphi(p) \lceil_{\tau_1}$
as well as $\varphi(q) \lceil_{\tau}$ are both $0$-trees.

We need sets of three assignments of size $t$ to be extendable
in  a locally consistent way and $3t \leq n/50$ ensures
that this is possible.
\end{proof}

When studying the number of lines in a proof where each line
is short, an extension of Lemma~\ref{lemma:noproof} is needed.
This was first done in \cite{PRT21} and we rely on
the argument of \cite{jhkr}.  

In this situation it is not possible to have
a single $t$-evaluation whose range is all the sub-formulas
that appear in the proof.  Instead, we have separate
$t$-evaluations for each line.  Furthermore, for each
line these $t$-evaluations live at the leaves of a decision
tree.  As these are very similar to what is called $\ell$-common decision
trees (as defined in \cite{jhmultiswitch}) we call them 
$\ell$-common $t$-evaluations.

\begin{definition}
A set of formulas $(F_i)_{i=1}^M$ has an $\ell$-common $t$-evaluation if
there is a decision tree of depth $\ell$ with
the following properties.  Take any leaf of this decision
tree and let $\tau$ be the partial assignment defined by
this path.  At this leaf we have a $t$-evaluation
$\varphi_{\tau}$ of the formulas $(F^i\lceil_{\tau}))_{i=1}^M$.
\end{definition}

We need an extension of Lemma~\ref{lemma:noproof}.
Suppose we are given a proof of $\PHP$ and for each line,
$\lambda$, let $\Gamma_\lambda$ be the set of all sub-formulas
of the formula appearing in this line and let $\varphi_{\tau}^{\lambda}$
be the $t$-evaluation given at the leaf defined by 
$\tau$ that contains $\Gamma_\lambda$ it is domain.

\begin{lemma}\label{lemma:noproofm}
Suppose we have a derivation using the rules
of Section~\ref{sec:frege} starting with the
axioms of the functional onto $\PHP$ on the $n \times n$ grid.
Suppose each $\Gamma_\lambda$ allows an $\ell$-common $t$-evaluation.
Suppose $\ell +t \leq n/300$, then
each line of the derivation is mapped to a $1$-tree
in all leaves of its decision tree.
\end{lemma}

The proof of this is not very complicated but needs some
notation.  The important property is
that the copies of the same formula appearing in multiple lines
are mapped in a consistent way.  Take a formula
$F$ appearing on lines $\lambda_1$ and $\lambda_2$.
Take any leaf in the common decision tree of 
$\lambda_1$ defined by an assignment $\tau_1$.
Take any assignment $\tau_1'$ consistent with $\tau_1$ in
the decision tree $\varphi_{\tau_1}^{\lambda_1} (F)$ leading
to a branch labeled $b$.  We claim the following.

\begin{lemma}\label{lemma:consteval}
For any $F$, $\tau_1$, $\tau_1'$, and $b$ as described above consider the
decision tree $T=\varphi_{\tau_2}^{\lambda_2} (F)$ 
where $\tau_2$ is locally consistent with $\tau_1$
and $\tau_1'$.  Then any path in $T$ defined
by an assignment $\tau_2'$ consistent with
$\tau_1$, $\tau_1'$ and $\tau_2$ leads to
a leaf labeled $b$.
\end{lemma}

\begin{proof}
We prove this by induction over the complexity of $F$.
The lemma is true when $F$ is a single variable
as this follows from the definition of a $t$-evaluation.

Now for the induction step.  When $F$ is of
the form $\neg F'$, then the the
lemma follows from the basic properties of
$t$-evaluations and the inductive case that
the lemma is true for $F'$.

Finally, assume that $F=\lor_{i=1}^m F_i$ and suppose
it violates the lemma getting a tuple $\tau_1, \tau_1',
\tau_2$, and $\tau_2'$ leading to leaves with
different values in the two trees.  Assume with
without loss of generality that $b=1$, the situation
being symmetric between the two lines.

By the definition of $t$-evaluations there
must be an $i$ such that 
$\varphi_{\tau_1}^{\lambda_1} (F_i) \lceil_{\tau_1'}$
is a $1$-tree.  Similarly, for
$\lambda_2$, we know that
$\varphi_{\tau_2}^{\lambda_2} (F_i) \lceil_{\tau_2'}$
is a $0$-tree.  As the four-tuple of restrictions is
consistent and $F_i$ is simpler than $F$ this
contradicts the inductive assumption.

Here we need two sets of pairs of assignments to be
locally consistent and hence $t+ \ell \leq n/100$
is sufficient.
\end{proof}

Lemma~\ref{lemma:consteval} tells us that
the same formula when considered on different
lines and under different $t$-evaluations
behaves the same.  The proof of Lemma~\ref{lemma:noproofm}
is now essentially identical to the
proof of Lemma~\ref{lemma:noproof}.  In this situation
we need triplets of three pairs of assignments to be
consistent and hence we use $t + \ell \leq n/300$.
We omit the details.

\section{Restrictions}\label{sec:restrictions}

As stated in the introduction we use a slightly more complicated object than a restriction
which normally only gives values to some variables.
A restriction in our setting fixes many variables to constants but
also substitutes the same variable or its negation for
some variables.  In a few cases an old variable is substituted
by a small logical formula which is a disjunction of
size at most three.

We are given an instance of the $\PHP$ on the $n\times n$ grid
and a restriction, for a suitable parameter $T$, reduces it to a smaller
instance on the $(n/T) \times (n/T)$ grid where $(n/T)$ is an odd number.

We divide the grid in to  $(n/T)^2$ squares, called super-squares,
each with side length $T$.  Inside each super-square 
there are $\Delta$ (for a parameter to be fixed) smaller squares that we from
now on call ``mini-squares'' (of side length larger than 1, but
smaller than the super-squares).  We pick one
mini-square inside each super-square and let these represent
the smaller instance.  Each super-square has a color as given
by a chess board coloring of the reduced instance.  For instance
the corner super-squares are all white.  Each mini-square has
the color of its super-square.

Between each mini-square, $s_i$ and any mini-square
$s_j'$ in an adjacent super-square we have $3R$ (for a parameter to
be chosen) edge-disjoint paths,
each of even length.  We can match exactly all vertices on any such
path by matching each node of the path to the appropriate adjacent
node.  We are also interested in matching each node to the other neighbor
on the path and in this case we need to include 
one node in each of the two mini-squares
to which it is attached creating a dent in the
sense of Lemma~\ref{lemma:matchcenter}.
We think of this as using
the path as an augmenting path and hence we sometimes refer to these
paths as ``augmenting paths''.  A matching on such a path is of
type $0$ if it does not include any node from the attached
mini-squares and otherwise it is of type $1$. 

We group
the $3R$ paths in groups of three and within each
group, two attach at a node which is the same color as
the color as the mini-square, while the third one
attaches to a node of opposite color.  
As each path is of even length, the nodes of
attachment at the two end-points are of different colors but so are the
mini-squares to which the path attaches.  Hence
the  attachment points are either both the same color as
the respective mini-square or both the opposite color.
The first two augmenting paths
may be of type $0$ or $1$.  The third path is always of
type $1$ and as these paths play little role in the argument
we mostly ignore them from now on.

Let us point out that for mini-squares in super-squares on the perimeter of the
entire grid, in some direction(s) there are no adjacent 
mini-square.  In this situation, no paths attach on the
side with no neighbor.

For each augmenting path $P$ we have a corresponding Boolean
variable $y_P$ which indicates whether it is of type $0$
or type $1$.  By the $R$ fixed paths of type $1$
we can conclude that for any mini-square, if
exactly half of its varying  paths are of type $0$ (and
hence the other half is type $1$) then its interior
has equally many white nodes as black
nodes remaining.  

We set up our restrictions such that it is uniquely
determined by the values for the variables $y_P$.
Outside the paths and
the mini-squares we more or less have a fixed matching.
This concludes the high level description of
a restriction and let us repeat the argument, now giving all
formal details.

\subsection{Details of mini-squares and paths}

It is convenient to
use the concept of a {\em brick} which is a square of size
$30 \times 30$.  We think of routing paths
through bricks, but in concrete terms such a path
is given by two points of attachment of different colors.
If the path is of type $0$ these two potential dents
are matched inside the brick and if the path is of type
$1$ these two points are matched to the appropriate neighbor
in the neighboring brick.   We make sure that the 
there are at most three dents in each side of the brick and
of distance at least $12$ to any corner and hence we can
use Lemma~\ref{lemma:matchcenter} to find a matching
of the interior.  This matching might not look like
a path going through the brick but from an abstract point
this is a good way to see it.

We have two parameters, $R$ and $\Delta$ and we
think of them as follows.  $R$ is needed to ensure
that of $R$ (almost) random bits it is likely that
about half of them are true.  For this reason $R=\Theta (\log n)$.
The parameter $\Delta$ is picked to make sure that 
$150 R \Delta^2$ is slightly less than $T$.

\paragraph{Mini-squares.} All but one mini-square have side length $120 \Delta R$
(or, equivalently, $4 \Delta R$ bricks).
We have $4\Delta R$ bricks on each side and 
the interesting part is the middle $2\Delta R$ bricks.
The bricks close to the corners are only used to 
ensure that we can perfectly match the interior of the
mini-square according to Lemma~\ref{lemma:matchcenter}.

Half of the middle bricks are used to route the $3\Delta R$ paths
to the $\Delta$ mini-squares in the super-square in the
given direction.  The other $\Delta R$ middle bricks
are used to make sure that there is enough space to
route the paths through different blocks between
the mini-squares.  

\paragraph{Designated survivor.} We have a special single mini-square in the top
left corner super-square\footnote{This could be any white
super-square and choosing this particular super-square is just to
make some fixed choice.}.   It has side length
$120 \Delta R +1$ which in particular is an odd number
but is otherwise like the other mini-squares.
We call this the ``designated survivor''.

\begin{figure} [h]
\begin{center}
  \scalebox{1.288}{
\begin{picture} (280,210)

\put(20,47){$T$}
\put(62,5){$T$}

\multiput(35,20)(60,0){4}
{\line(0,1){180}}
\multiput(35,20)(0,60){4}
{\line(1,0){180}}

\put(36,20){\line(0,1){174}}
\put(41,199){\line(1,0){174}}

\multiput(35,74)(0,60){3}
{\line(1,0){12}}
\multiput(41,68)(0,60){3}
{\line(1,0){12}}
\multiput(47,62)(0,60){3}
{\line(1,0){12}}
\multiput(53,56)(0,60){3}
{\line(1,0){12}}
\multiput(59,50)(0,60){3}
{\line(1,0){12}}
\multiput(65,44)(0,60){3}
{\line(1,0){12}}
\multiput(71,38)(0,60){3}
{\line(1,0){6}}

\put(35,198){\line(1,0){6}}
\put(35,196){\line(1,0){6}}

\put(37,200){\line(0,-1){6}}
\put(39,200){\line(0,-1){6}}

\put(44,188){\line(0,-1){37}}
\put(44,151){\line(1,0){18}}
\put(62,151){\line(0,-1){35}}

\put(56,176){\line(0,-1){30}}
\put(56,146){\line(-1,0){6}}
\put(50,146){\line(0,-1){18}}

\put(44,128){\line(0,-1){35}}
\put(44,93){\line(1,0){12}}
\put(56,93){\line(0,-1){31}}

\put(47,71){\line(1,0){40}}
\put(87,71){\line(0,-1){12}}
\put(87,59){\line(1,0){26}}

\put(113,125){\line(1,0){40}}
\put(153,125){\line(0,-1){12}}
\put(153,113){\line(1,0){26}}

\put(119,119){\line(1,0){25}}
\put(144,119){\line(0,-1){4}}
\put(144,115){\line(1,0){35}}

\multiput(41,80)(0,60){3}
{\line(0,-1){12}}
\multiput(47,74)(0,60){3}
{\line(0,-1){12}}
\multiput(53,68)(0,60){3}
{\line(0,-1){12}}
\multiput(59,62)(0,60){3}
{\line(0,-1){12}}
\multiput(65,56)(0,60){3}
{\line(0,-1){12}}
\multiput(71,50)(0,60){3}
{\line(0,-1){12}}
\multiput(77,44)(0,60){3}
{\line(0,-1){6}}

\multiput(95,74)(0,60){3}
{\line(1,0){12}}
\multiput(101,68)(0,60){3}
{\line(1,0){12}}
\multiput(107,62)(0,60){3}
{\line(1,0){12}}
\multiput(113,56)(0,60){3}
{\line(1,0){12}}
\multiput(119,50)(0,60){3}
{\line(1,0){12}}
\multiput(125,44)(0,60){3}
{\line(1,0){12}}
\multiput(131,38)(0,60){3}
{\line(1,0){6}}

\multiput(101,80)(0,60){3}
{\line(0,-1){12}}
\multiput(107,74)(0,60){3}
{\line(0,-1){12}}
\multiput(113,68)(0,60){3}
{\line(0,-1){12}}
\multiput(119,62)(0,60){3}
{\line(0,-1){12}}
\multiput(125,56)(0,60){3}
{\line(0,-1){12}}
\multiput(131,50)(0,60){3}
{\line(0,-1){12}}
\multiput(137,44)(0,60){3}
{\line(0,-1){6}}

\multiput(155,74)(0,60){3}
{\line(1,0){12}}
\multiput(161,68)(0,60){3}
{\line(1,0){12}}
\multiput(167,62)(0,60){3}
{\line(1,0){12}}
\multiput(173,56)(0,60){3}
{\line(1,0){12}}
\multiput(179,50)(0,60){3}
{\line(1,0){12}}
\multiput(185,44)(0,60){3}
{\line(1,0){12}}
\multiput(191,38)(0,60){3}
{\line(1,0){6}}

\multiput(161,80)(0,60){3}
{\line(0,-1){12}}
\multiput(167,74)(0,60){3}
{\line(0,-1){12}}
\multiput(173,68)(0,60){3}
{\line(0,-1){12}}
\multiput(179,62)(0,60){3}
{\line(0,-1){12}}
\multiput(185,56)(0,60){3}
{\line(0,-1){12}}
\multiput(191,50)(0,60){3}
{\line(0,-1){12}}
\multiput(197,44)(0,60){3}
{\line(0,-1){6}}

  \end{picture}}
\caption{The placement of mini-squares and super-squares and some
paths.  The designated survivor is the checkered mini-square.
The matchings along top row and leftmost column are indicated
by solid lines.
 \label{fig:grid}}
\end{center}
\end{figure}

\paragraph
{The first row and first column.} 
In the top row, the $n-120 R \Delta -1$ nodes outside
the designated survivor are matched in a horizontal
matching.  Similarly, the nodes in the leftmost column
outside the designated survivor are matched in a 
vertical matching.  This essentially eliminates one
row and one column outside the top left
super-square and we assume for convenience that $n\equiv 1$
modulo $30$ and we cover the rest of the grid with bricks.

\paragraph{Structure of super-squares.}
A super-square
is a square of bricks of side length $5 \Delta^2 R$ bricks.
We need $150 \Delta^2 R$ to be slightly smaller than $T$.
We might need a few extra rows and columns for divisibility
reasons but for notational simplicity we ignore this difference
and simply assume $T=150 \Delta^2 R$.
Starting in the top left corner in each super-square
we have $\Delta$ mini-squares
along the diagonal.  This leaves $\Delta^2 R$ empty rows
of bricks in the bottom of each super-square and $\Delta^2 R$ empty
columns at the right.
By our placement of the designated survivor and the elimination
of the top row and first column also the top left super-square
looks essentially the same as other super-squares. 

\paragraph{Structure of paths.}
Fix a mini-square $s_i$ in a super-square $S$ and let us see how to
route paths to mini-squares $s_j'$ in the super-square to its right.  For
each pair $(i,j)$ we reserve $R$ columns, $c_{ij}^k, \ 1 \leq k \leq R$, 
of bricks in the right part of $S$.  This can be done has
we have $\Delta^2$ pairs of mini-squares and $\Delta^2 R$
empty columns of bricks.

We number the middle bricks on each side from
one to $2 \Delta R$ and for each $j$ we reserve $R$ bricks, 
all with even number
on the right perimeter of $s_i$.  The $k$th brick contains
three paths that we route straight right to $c_{ij}^k$.
Similarly we route paths straight left from the $k$th brick 
in the $i$th part of $s_j'$ to the same column.  This time
using odd numbered bricks.  The path is completed
by using the suitable part of $c_{ij}^k$ to connect the
two pieces.  

Connecting mini-squares vertically is done completely analogously
and we omit the description.  Let us make a not very difficult observation.

\begin{lemma}\label{lemma:disjoint paths}
In each brick outside the mini-squares there is
at most one set of three vertical paths and at most
one set of three horizontal paths.
\end{lemma}

\begin{proof}
Let us look at the area to the right of the diagonal in one
super-square and to the left of the diagonal in the super-square
to its right.  The only vertical parts of paths in this area
are from the middle segments of the paths.  By definition
these are in disjoint columns.

The horizontal parts of paths in this area  are from either segment one
or segment three.  As the former are in even numbered rows
of bricks and the latter in odd numbered rows these
are also disjoint.

The argument for the area below a diagonal of one super-square and above
the diagonal below, is completely analogous.

The remaining part is the area down to the right in each mini-square, i.e.
in the intersection of columns with no mini-squares and 
rows with no mini-squares.
In these bricks there are no paths.
\end{proof}

We summarize the properties we need from
this construction.  

\begin{lemma}\label{lemma:newmatching}
The values of the path variables $y_P$ uniquely determines
a matching on all paths and in any brick outside
the mini-squares.  In any mini-square except
the designated survivor such that half of its adjacent
$y_P$ variables are true we can find a unique matching
of the remainder of this mini-square.
In the designated survivor we can find a matching
provided one more than half the adjacent
variables is true.
\end{lemma}

\subsection{Almost complete matchings and restrictions}

We let an {\em almost complete matching}, usually denoted $\tau$,
be an assignment to all variables $y_P$ such that exactly half the
variables next to any mini-square are true.  Note that many such
$\tau$ do exist and in particular we can pick half the paths
in any group of $2R$ paths to be of each type.
There are many other ways to pick $\tau$ but we
do not need this explicitly.

Our restrictions are constructed by giving
fixed values to some variables $y_P$ while
some remain unset.  We rename the latter $z_P$ to distinguish
new and old variables.  While an individual $x_e$ might depend on
many $y_P$ it only depends on at most three new variables.

Let us proceed to describe how to pick a random restriction from
our space.  We make uniformly 
random choices but if some choice is very unlucky we redo
this choice.   Let $\pi_1$ be a fixed matching of
all super-squares except the top left square where
the designated survivor is located.  We pick a restriction $\sigma$
as follows.  

\begin{enumerate}

\item Pick a uniformly random almost complete matching $\tau$.
If any group of $2R$ variables between two fixed
mini-squares has fewer than $R/2$ variables
of either value, restart.
We denote this event as ``lopsided group''.

\item Pick a random mini-square from each super-square, except
the super-square of the designated survivor forming, jointly
with the designated survivor, the set $U$.
Match these mini-squares according
to $\pi_1$.  Call the resulting matching (now of
mini-squares) $\pi_1^U$.

\item For each pair of mini-squares $(s_1,s_2)$ matched in $\pi_1^U$
pick one augmenting path of type $1$ and convert it to type $0$.
These are called chosen paths.  The choice is based on
an advice string $B$ as discussed in Section~\ref{sec:change}
below.

\item For each pair of mini-squares $(s_1,s_2)$ in $U$ in
adjacent super-squares but not matched in $\pi_1^U$
pick one augmenting path of type $0$ and make it a chosen path.
The choice of the path is based on the advice string as 
discussed in Section~\ref{sec:change} below.
This is also done with $s_1$ being the designated survivor.

\end{enumerate}

For each chosen path $P$ we have corresponding variable
$z_P$.  These give the variables in the reduced instance.

\subsection{The reduced instance}

The nodes of the new instance are given by the elements
in $U$.
They naturally form a $(n/T) \times (n/T)$ grid.
We have chosen paths that are of type $0$ that connect
any two adjacent elements in $U$ while for any other
path $P$, the value of $y_P$ is now fixed. 

By Lemma~\ref{lemma:newmatching} if exactly one
variable corresponding to a chosen path next to a mini-square in $U$ is true,
then it is possible to find a matching of this mini-square.
Thus the local conditions of a mini-square turn in to
an axiom of the new instance of $\PHP$.
Let us see how to replace the old variables
in a supposed proof with these new variables.
First note that old variables not on edges in chosen mini-squares
or in bricks with at least one chosen path
are now fixed in a way respecting the corresponding
axioms.

We now describe how to restrict the variables of $\PHP_n$
according to a full restriction $\sigma$.
Consider a brick with at least one chosen path
going through.  There is only one chosen path
between any two adjacent chosen mini-squares and it is not difficult to
see that at
most one chosen path goes through any brick.
There are two possible matchings of this brick
depending on the value of the corresponding
variable $z_P$.  If an edge $e$ is present in
neither we replace $x_e$ by $0$ and if it is
present in both we replace it by $1$.
If $e$ is present in only one we replace it
by $z_P$ or $\bar z_P$ in the natural way.
It is easy to see that any axiom
related to a node in the brick becomes
true.

Similarly, in a mini-square we have
four (or fewer if it is on the perimeter) chosen
paths next to it. These are controlled by four
new variables that we here locally call
$z_i$ for $1 \leq i \leq 4$.  The new local
axiom is that exactly one of these four
variables is true.

We have four different matchings of the mini-square
depending on which $z_i$ is chosen to be true.
Look at an edge, $e$ and suppose, for example, that it appears 
in the matchings corresponding to $z_2$ and $z_3$
being true.  In such a case we replace $x_e$ by $z_2 \lor z_3$ and
similarly in other cases.  If $e$ is in none of the
four matchings we replace $x_e$ by the constant $0$ and
if it is in all four we replace it by
the constant $1$.

It is easy to check that any original axiom inside
the mini-square either reduces to true or that
exactly one of the four $z_i$ is true.
Indeed, looking at the disjunctions replacing
the four variables $x_e$ around any node, each $z_i$ appears
in exactly one.  We summarize the discussion of
this section as follows.
\begin{lemma}\label{lemma:reduced}
A restriction $\sigma$ reduces an axiom in the 
$\PHP_n$ either
to the constant true or an axiom in the $\PHP_{(n/T)}$ of
the new variables $z_P$.
Each variable $x_e$ is substituted by a conjunction
of up to three literals.
\end{lemma}

Let us finally note that each variable $x_e$ that is
not replaced by a constant or a single variable
appears inside a chosen mini-square corresponding
to a node, $v$, in the new instance.  It is completely
determined by the answer to the question ``To which node
is $v$ matched'' and hence it is not very different from
other variables.  In particular since it is determined by
one basic query there is no problem with a build-up
of complicated expressions when we compose restrictions.

As in previous papers these full restrictions
are not used in the main argument and we work
with partial restrictions that we now turn to.
The intuitive reason to introduce these 
is that a full restriction is a very rigid object with exactly
one live mini-square in each super-square.  In the final
argument we compare the number of restrictions to the
number of restrictions with slightly less live variables.
This does not work for full restrictions as the number
of full restrictions that have $s$ live mini-squares 
removed is actually higher than the number of full
restrictions.   If we instead first globally pick
a larger number of mini-squares to keep alive,
then decreasing the number of such live mini-squares
does decrease the number of alternatives. Let us
turn to defining these partial restrictions.

\subsection{Partial restrictions}

Let $k$ be a parameter equal to $C (n/T)^2 \log n$
for a sufficiently large constant $C$.  After
we have completed the construction of $\sigma$
we add the following steps.  \label{distribution}

\begin{enumerate}

\item Pick, without replacement of mini-squares, $k$ uniformly
random pairs of mini-squares in adjacent super-squares.  
These are picked one at the time and each is picked uniformly
from the set of remaining possible pairs.
This yields a matching $\pi_2$.  If any super-square   
has more than $4C \log n$ live mini-squares
we consider the choice unbalanced and we restart.
Also, if for any two adjacent super-squares $S_1$ and $S_2$ 
if we have less $C \log n/4$ 
pairs $(s_1,s_2)$ with $s_i \in S_i$ we consider the choice
unbalanced and
restart.

\item Change the type of one augmenting path between
any pairs of nodes in $\pi_2$ from $1$ to $0$.  The choice of which of
the at least $R/2$ paths of type $1$ is based
on the advice string as discussed in Section~\ref{sec:change}
below.

\end{enumerate}

The mini-squares picked during this process jointly with
the chosen mini-square are considered ``alive''.
For any path, $P$ between two live mini-squares the
corresponding value, $y_P$, is now considered undetermined while other
$y_P$ variables are fixed.  When we later go from
a partial restriction to a full restriction it is not true that
these undetermined variables can vary freely as in fact for 
each pair of live mini-square we may change the value of
at most one of its adjacent paths.  In spite of this
we do not consider such a value to be known unless
we have the full information at one of its end points.

As we pick $2k$ mini-squares we
expect roughly\footnote{The reason this is not exactly true
is super-squares at the perimeter have only two or three neighboring
super-squares.  Such squares are less likely to have many live mini-squares.  
This results in a factor $(1+o(1))$ more mini-squares in
other super-squares but this small factor does not matter and
we ignore it.}
$2C \log n$ live mini-squares in any super-square and hence,
by standard Chernoff bounds it is unlikely that this number is
larger than $4C \log n$ for any super-square.
Similarly, for any two adjacent super-squares $S_1$ and $S_2$ we expect
$C \log n /2$ pairs $(s_1,s_2)$ picked such that $s_i \in S_i$
and hence it is unlikely that this number is smaller than 
$C \log n/4$ for any pair $(S_1,S_2)$.

We call the resulting restriction $\rho$.  It determines values
for variables $y_P$ exactly as for $\sigma$.  
A variable $x_e$ is fixed to a constant
unless it as influenced by at  least one live mini-square.  Either from
being within such a mini-square or in a brick with at least one
live path passing through.

\subsection{Changing types of augmenting paths} \label{sec:change}

In the above procedure, in two places we need to select
an augmenting path and (possibly) change its type.
This happens when changing a path from type $1$ to type $0$ because its end-points are matched in $\pi_1^U$
or $\pi_2$ and when opening up for changing the type from
$0$ to $1$ by making it a chosen path.

We could accept to make this choice arbitrary by losing
some factors of $R$ in our bounds, but as it is always
nice to avoid unnecessary loss let us describe a
more efficient choice.

The choices of the $2R$ variables in
a group corresponds to a vector in $\{ 0, 1 \}^{2R}$ and
we want to modify one coordinate in order to change the Hamming
weight from $t$ to either $t+1$ or $t-1$ and let us suppose the latter.
We want the choice to be limited and as invertible
as possible.  As the number of strings of weights $t$ and
$t-1$ are different we cannot achieve perfect
invertability.   Suppose first that $t \leq R$.

\begin{definition}
Suppose $t \leq R$.  A mapping $f$ mapping 
${{2R} \choose t}$ to ${{2R} \choose {t-1}}$ which
maps each set to a subset,
is a $k$-{\em almost bijection} if
it is surjective and $| f^{-1} (x)| \leq k$ for
any $x$.  
\end{definition}

The following below lemma is probably well known but as the
proof is not difficult, we prove it.  It is likely
that 4 can be improved to 3 but this does not matter greatly for
us, as this only affects unspecified constants.

\begin{lemma}\label{lemma:restrictedchoice}
If $R/2 \leq t \leq R$ then there is a
4-almost bijection.
\end{lemma}

\begin{proof}
Consider a bipartite graph where the left
hand side elements are subsets of size $t-1$ and the right
hand side elements are subsets of size $t$.  Connect
two sets iff one is a subset of the other.  It is well
known (see for instance Corollary~2.4 in \cite{bollobas})
that this graph has
a matching, $M_1$, of size $2R \choose t-1$.

Modify the construction by making three copies of 
each left hand side node.  Each copy is again connected
to any set that contains it.  It follows by
the LYM inequality (stated as Theorem~3.3 in \cite{bollobas})
that this graph has
a matching $M_2$ of size $2R \choose t$.

Now define $f(x)$ as follows.  If $x$ is matched
in $M_1$ let it be its partner in this matching.
If $x$ is not matched in $M_1$ define $f(x)$ to be the
partner under $M_2$.

Due to the first condition $f$ is onto.  The property
that $| f^{-1} (y)| \leq 4$ for follows as
a preimage of $y$ is either its partner under $M_1$
or a partner of one of its three copies under $M_2$.
\end{proof}

Taking the complement of both input and output
we define a $4$-almost bijection, $g$ mapping
$2R \choose t$ to $2R \choose {t+1}$ for $R \leq t \leq 3R/2$.
We use $f$ and $g$ to guide our choices and
in addition we have
two bits of advice for each group.

\begin{definition} \label{def:advice}
For each group of $R$ paths we have two bits in the 
advice.
\end{definition}

When we want to convert a path from type $1$ to type
$0$ between $s_i$ and $s_j'$ we look the types of all paths between
the two mini-squares.  This is a vector, $v$, in $\{ 0,1\}^{2R}$
which by the non-lopsidedness has Hamming weight
$t$ which is in the interval $[R/2,3R/2]$.
If $t> R$ we look at $g^{-1}(v)$ and consider the
two bits of advice, $b_1$ and $b_2$.  All we
need to do is to ensure that each choice
in $g^{-1}(v)$ is possible but to be explicit we
can proceed as follows.

\begin{itemize}

\item If $g^{-1}(v)$
is of size one we pick the unique element.  

\item If $g^{-1}(v)$
is of size two we use $b_1$ to make the choice.

\item If $g^{-1}(v)$
is of size three then $b_1=b_2$ we pick
the lexicographically first path and
otherwise we use $b_1$ to choose between
the other two paths.

\item If $g^{-1}(v)$
is of size four then use the pair $(b_1,b_2)$ to
make the choice.

\end{itemize}
The reason for the advice string is to get a pure counting
argument when we later analyze probabilities. It would
have worked to pick a random element from
$g^{-1}(v)$.

To make the situation uniform we have two
advice bits for any pair of mini-squares in adjacent
super-squares.  One can note that most of these bits are never used
but they make the construction uniform.
We let $B$ denote the values of all these bits.

If $t \leq R$ we instead consider $f(v)$ and change the type
of the corresponding path.  Finally if we want
to change the weight from $t$ to $t+1$ we reverse the two 
cases.

\subsection{Analyzing the probability of a restart}

We make a restart either because of a lopsided
group or an unbalanced pick of $\pi_2$ and we analyze
these separately.  We start with lopsided groups.

\begin{lemma}\label{lemma:notfull}
The probability that uniformly random 
$\tau$ has lopsided group is $O(n^2 2^{-cR})$
for a positive constant $c$.
\end{lemma}

Based on this lemma we fix $R$ to be $C \log n$ for
a sufficiently large constant $C$ such that the probability of 
having a lopsided group is $o(1)$.  Let us
prove Lemma~\ref{lemma:notfull}.

\begin{proof}
The almost complete matching $\tau$ is defined by the 
variables $y_P$ which we in this section choose
to take values $1$ and $-1$.  For each mini-square, sum the variables
on its boundary and for a uniformly random assignment to
all variables, let $Z$ be the vector of all these mini-square sums.
Let us denote the number of mini-squares by $d$ making
$Z$ an integer  vector of length $d$. 
When constructing $\tau$ we are conditioning on the event $Z=0^d$.

Fix any two mini-squares $s_1$ and $s_2$ and let
$g$ be the group of $2R$ variables associated with
paths between $s_1$ and $s_2$.
Let $X_g$ be sum of the variables
in this group.  We want to estimate
the probability that $X_g= x$ where $x$ is
either at most $-R$ or at least $R$.  Let $Z'_g$ be the set of mini-square
sums when the paths between $s_1$ and $s_2$ are
removed.  As these two mini-squares also have other
adjacent augmenting paths this is still a vector of
length $d$.  Let $v_x$ be the
vector of length $d$ that has $x$ at positions $s_1$ and $s_2$
and is otherwise $0$.  We want to estimate
$$
Pr[X_g=x \ |\ Z=0^d]=Pr[X_g=x \land Z=0^d]/Pr[Z=0^d]
$$
which equals
\begin{eqnarray*} 
 Pr[X_g=x \land Z_g'= -v_x]/Pr[Z=0^d]
\end{eqnarray*}
and as the two events are independent this equals
$$
Pr[X_g=x] Pr[Z_g'= -v_x]/Pr[Z=0^d].
$$
We have the following
lemma of which we postpone the proof.

\begin{lemma}\label{lemma:max0}
For any outcome $v \in \mathbb{Z}^d$ we have
$Pr[Z_g'=v] \leq Pr[Z_g'=0^d]$.  
\end{lemma}

In view of the lemma we get the upper bound
$$Pr[X_g=x] Pr[Z_g'= 0^d]/Pr[Z=0^d]
$$
for the probability we want to estimate.
Clearly $Pr[Z=0^d] \geq Pr[X_g=0] Pr[Z_g'= 0^d]$
and substituting this into the equation we get
the upper bound $Pr[X_g=x] / Pr [X_g=0]$.
When $|x| \geq R$, then by standard Chernoff bounds,
this probability is $2^{-cR}$
for some explicit $c$ and since there are at
most $n^2$ pairs of mini-squares the lemma follows.
\end{proof}

Let us prove Lemma~\ref{lemma:max0}.
\begin{proof}[Proof of Lemma~\ref{lemma:max0}]
Let $f(v)$ be the probability that $Z'_g=v$.
As $v$ is the vector sum of contributions of single
$y_P$,  $f(v)$ is a giant
convolution.  To be more precise for each path $P$ between
mini-squares $s_1$ and $s_2$ we have a vector $v_1$ with
a one in positions corresponding to $s_1$ and $s_2$ and
0 in all other positions.  Define a probability distribution
$f_P$ that gives probability $\frac 12$ to each of
$v_1$ and $-v_1$.  The function $f(v)$ is
the convolution of $f_P$ over all paths $P$.
If $\hat f_P$ is the Fourier transform of $f_P$ then
the Fourier transform of $f$ is 
\begin{equation}\label{eq:fourier}
\hat f(x)= \prod_P \hat f_P(x),
\end{equation}
where $x$ belong to the $d$-dimensional torus.
Note that as $f_P$ is symmetric around $0^d$, $\hat f_P$ is
real-valued.  Moreover,  for each pair of neighboring
mini-squares $s_1$ and $s_2$ we have $2R$ paths between
$s_1$ and $s_2$.  This implies that the right-hand side
of (\ref{eq:fourier}) can be written as
$$
\prod_{s_1,s_2} \hat f_{s_1,s_2}^{2R},
$$
where $f_{s_1,s_2}=f_P$ for any path $P$ with endpoints
$s_1$ and $s_2$.   We conclude that $\hat f$ only
takes real and non-negative values.  We conclude that,
as for any function with a real-valued and non-negative Fourier transform,
$f(0)\geq f(v)$ for any $v$ and this is exactly what we wanted to 
prove.
\end{proof}

Let us next discuss the balance condition when
picking $\pi_2$.

\begin{lemma}\label{lemma:unbalance}
The probability that $\pi_2$ is unbalanced
is $O(n^{-2})$ provided $C> C_0$ for some
fixed constant $C_0$.
\end{lemma}

\begin{proof}
This might follow
from the fact that Chernoff bounds are true for
negatively correlated variables.  It is, however,
not quite obvious that
our variables are negatively correlated.  It is true that
fixing the size of the matching causes negative correlation
but the requirement to be a matching is more complicated
and could give positive correlations.  In view of this
let us sketch a direct argument.

If we picked the pairs of mini-squares
with replacement the lemma would be completely standard.  
Let us analyze the dynamic process. To see
that we do not pick more than $4C \log n$ mini-squares
in any super-square with high probability we note two
facts.

\begin{itemize}

\item As we only pick an $o(1)$ fraction of all mini-squares,
at each point in time a fraction $(1-o(1))$ of all pairs
are available.

\item In view of this the probability that any
single mini-square is picked is only a $(1+o(1))$
factor larger compared to the procedure 
with replacement.

\end{itemize}
That we are unlikely to pick many mini-squares in a single
super-square now follows from the corresponding result
for the process of picking with replacement.

We turn to the condition that we have at 
least $C \log n /4$ pairs in any two adjacent super-squares.
Also, this analysis is completely standard if edges are picked with
replacement.  If we condition on not picking more than $4C \log n$
mini-squares in any super-square the probability that a picked edge is between
two given super-squares does not decrease by more than a factor
$1-o(1)$.  Hence the probability of picking very few edges
between two given super-squares in the process without replacement 
is not so different compared to the probability of
the same event in the process with replacement.
We leave it to the reader to fill in the details.
\end{proof}

\section{The switching lemma} \label{sec:switch}
In this section we establish the following basic
switching lemma.

\begin{lemma} \label{lemma:switch}
There is a constant $A$ such that the
following holds.
Suppose there is a $t$-evaluation that includes $F_i, 1\leq i \leq m$
in its range and let $F =  \lor_{i=1}^m F_i$.
Let $\sigma$ be a random restriction from the 
space of restrictions defined in Section~\ref{sec:restrictions}.
Then the probability that $F\lceil_\sigma$
cannot be represented by a decision tree of depth at most $2s$
is at most
$$
\Delta (A(\log n)^{3} t \Delta^{-1})^{s}.$$
\end{lemma}

\begin{remark}
A good way to think about the parameters when applying 
Lemma~\ref{lemma:switch} is that $R$ is given
by the choice of $n$.  We then choose $\Delta$ to
be sufficiently large to get a small failure probability.
The value of $T$ that controls how much the instance
shrinks is then calculated as $T= 150 \Delta^2 R$,
and this also determines the value of $k$.
\end{remark}

\begin{proof}
We are interested
in a $\sigma$ that gives a long path in the decision tree.
As in previous papers \cite{jhtseitin}, \cite{jhkr}
we explore what has been called ``the canonical decision tree''
 under the partial  restriction $\rho$ which
we from now on call simply ``a restriction'' dropping
the word ``partial''.  As $\sigma$ has fewer live
variables compared to $\rho$ this is sufficient to
establish the lemma.

For any variable $x_e$ we define its {\em influential mini-square(s)}.
This is either a single mini-squares or two mini-squares.
We want the property that if we know the values of all $y_P$
around the influential mini-square(s) then this uniquely
determines the value of $x_e$.
If $e$ is within a mini-square then this mini-square
is its influential mini-square(s).  If $e$ is in a brick outside
the mini-squares then the influential mini-squares are the
closest end point(s) of the live path(s) in this brick.
By ``closest'' we here mean a live mini-square such that
if we know the value of all adjacent variables $y_P$ then
we know the value of the $x_e$.

\begin{definition}
The value of $x_e$ is {\em forced} iff the values
of $y_P$ is known for all paths connected
to its influential mini-square(s).
\end{definition}

If we are interested in the value of a variable $x_e$ that
is not forced, we need to find out more information to
change this state of affairs.  The main way to get such information is
by something we call matched pairs.

\begin{definition}
If two mini-squares $s_1$ and $s_2$ give a {\em matched pair} 
then we should change the value of $y_P$ from $0$ to
$1$ where $P$ is a path with end-points
$s_1$ and $s_2$.
\end{definition}

\begin{remark}
We use this concept for pairs of neighboring chosen mini-squares
and for pairs in the matching $\pi_2$.  In both cases
the identity of the path to change follows from the context.
For pairs of chosen mini-squares it is the chosen path
and for pairs in $\pi_2$ it is the path the was modified
when it was first picked.
Note that a matched pair $(s_1,s_2)$
fixes, once and for all, the values of all variables $y_{P'}$,
 where $P'$ is a path with an end-point $s_1$ or $s_2$
\end{remark}

We now proceed to define the
canonical decision tree.  Remember that while our starting trees
have variables corresponding to edges in the grid, the tree we 
are creating has variables corresponding to the chosen paths.

The process of creating the canonical decision tree
is guided by $\rho$
and a set $I$ of matched pairs.  The support of the set
$I$ is the mini-squares appearing in any matched
pair and this is initially empty.

Letting $T_i$ denote $\varphi(F_i)$,
we go over the branches of $T_i$ for increasing
values of $i$.  As the values to some variables
are forced by the current information we cannot freely
follow any branch.  We locate the first (in any fixed order)
branch that leads to one and which can be followed
by the current information and call it the {\em forceable} branch.
Let us be formal.

Before stage $j$ we have an information set $I^{j}$ in the form of
some matched pairs.  It contains
some pairs from $\pi_2$ and some pairs based on answers in the decision
tree.  In stage $j$ we have forcing information $J_j$ that forces
the values of all variables on the forceable branch leading
to a one.  This set contains.

\begin{enumerate}

\item A set of matched pairs from $\pi_2$.

\item A set of matched pairs of chosen mini-squares, consistent
with the information set $I^j$.

\end{enumerate}
By ``consistent'' we mean that the resulting partial
matching on chosen mini-squares is locally consistent in the 
sense of Definition~\ref{def:consistent}.

We now extend the canonical decision tree
by, for any chosen mini-square mentioned in $J_j$ we ask
for its partner in the decision tree.  This
information, jointly with the matched pairs from $\pi_2$
in $J_j$, forms the $j$th information set, $I_{j}$, and
we set $I^{j+1}= I^j \cup I_j$.

Given $I^{j+1}$ we can determine whether the forceable
branch is followed.  If it is, we answer $1$ in the canonical
decision tree and halt the process.  Otherwise, we go to the next stage
and look for the next forceable branch.  If there is no more forceable
branch we halt with answer $0$.  Let $T$ be the resulting decision
tree.  To see that this is an acceptable choice for $\varphi (F)$,
we have a pair of lemmas.

\begin{lemma}
Let $\gamma$ be a set of answers in the decision
tree.  If there is no forceable branch given this
information, then, for each $i$,
$T_i \lceil_{\sigma \gamma}$ is a $0$-tree.
\end{lemma}

\begin{proof}
Suppose there is a locally consistent branch in $T_i \lceil_{\sigma \gamma}$ 
that leads to a leaf labeled one.  
The information
used to follow this branch can be used as forcing
information.
\end{proof}

\begin{lemma}
Let $\gamma$ be a set of answers in the decision
tree.  If we answer $1$ then there is an $i$ such that
$T_i \lceil_{\sigma \gamma}$ is a $1$-tree.
\end{lemma}

\begin{proof}
In fact for $T_i$ used for the construction of 
the forceable branch we have now reached a leaf that is labeled $1$.
\end{proof}

We use Razborov's labeling argument \cite{Razborov1995}
to analyze the probability that we make $2s$ queries in the 
canonical decision tree.  Slightly oversimplifying, given
a $\rho$ that gives a long path in the canonical decision
tree, we create a $\rho^*$ with fewer live centers such
that we can recover $\rho$ from $\rho^*$ and some
limited size external information.

Take any branch of this length
and suppose it was constructed during $g$ stages
using information sets $J_j$.  
Let $J^*=\cup_{j=1}^g J_j$ and as any query is a result of including
an element in $J^*$ we know that it contains 
at least $s$ pairs and for notational
convenience we assume that this number is exactly $s$.
As the estimate for the probability of getting a set $J^*$
of size $s+d$ is exponentially decreasing in $d$ we can
sum over all values of $d$, formally justifying this assumption.

Note that $J^*$ may not be locally consistent when seen as
a partial matching on the chosen nodes, but this is not required.
We proceed to analyze the probability that the process
results in a $J^*$ of size $s$.  Let us start with
an easy observation.

\begin{lemma}\label{lemma:disjoint}
The support sets of $J_j$ are disjoint.  The support of
$J_j$ is also disjoint with the support of $I_i$ as
long as $i\not= j$.
\end{lemma}

\begin{proof}
The parts coming from $\pi_2$ are clearly disjoint
as $\pi_2$ is a matching.  The pairs of chosen nodes
are also disjoint as any mini-square included in $J_j$
is included in $I_j$ and later $J_{j'}$ are disjoint
from $I_j$.
\end{proof}

In the spirit of \cite{jhtseitin} and \cite{jhkr} we want to find a restriction
$\rho^*$ with fewer live mini-squares and then show how to
reconstruct $\rho$ from $\rho^*$ and some external information.
First note that it is possible for several tuples $(\tau, U, \pi_2, B)$
to produce the same $\rho$ and hence to make counting unambiguous we
count such 4-tuples and not partial restrictions.  In the same vein
we do not reconstruct it from $\rho^*$ but rather from a
quadruple $(\tau^*, U^*, \pi_2^*, B^*)$ (and some external information).
The quadruple 
$(\tau^*, U^*, \pi_2^*, B^*)$ does produce a restriction $\rho^*$ but
that is not what we count.  As the value of $\pi_1$ is fixed,
$\pi_1^U$ is determined by $U$ and hence we do not need to include
$\pi_1$ in the information.

Informally we 
want $\rho^*$ to be $\rho$ modified by the information set $J^*$.
In other words for a matched pair in $J^*$ we change the type of one augmenting
path between the two mini-squares and now consider the two
endpoints to be dead.  We proceed to find a quadruple that
gives this restriction.  The main crux is to construct $\pi_2^*$
and $U^*$ and let us start with this.

We initialize $\pi_2^*$ to  $\pi_2$ with all matched pairs in $J^*$ 
removed and $U^*=U$.  
We process the matching pairs, $(s_1,s_2)$ of chosen nodes
in $J^*$ one by one in the order as they appear in $J_j$.
If there is a pair $(s_1', s_2')$
in $\pi_2^*$ where $s_i'$ is in the same super-square as $s_i$ we remove
this pair from $\pi_2^*$ and replace $s_i$ by $s_i'$ in $U^*$ for $i=1$ and
$2$.  If there is no such edge we allow for two ``holes'' of
two empty super-squares in $U^*$ as we still remove $s_1$ and $s_2$.

Next we construct $\tau^*$.  There is not any real choice
how to do this as we want $\rho^*$ to look like $\rho$ with the
additional information given by $J^*$ but let us go over the 
details.   For any square already dead in $\rho$ nothing changes
and all information stays the same.  For the live nodes
there are a number of cases.  One important property
to keep in mind is that for an alive but not chosen
node the pairings are the same in $\pi_2$, $J^*$ and their
$I_j$ sets. When it comes to the chosen nodes this is not the case
and they can be matched to different nodes in $\pi_1^U$,
$J^*$ and their $I_j$ sets.
Let us consider a mini-square $s$ in super-square $S$.
We have a few alternatives.

\begin{enumerate}

\item The mini-square $s$ is contained in a matched pair from $\pi_2 \cap J^*$. \label{case:p2in}

\item The mini-square $s$ is contained in a matched pair from $\pi_2$ and 
moved to $U^*$. \label{case:p2out}

\item The mini-square $s$ is contained in a matched pair from $\pi_2$ and 
is not moved to $U^*$ and does not belong to $J^*$. \label{case:p2outnot}

\item The mini-square $s$ contained in $U$ and neither it nor its
partner in $\pi_1^U$ takes part in the above process. \label{case:p1not}

\item The mini-square $s$ is contained in $U$ and not in a matched
pair in $J^*$, but where its partner in $\pi_1^U$ is contained in 
a pair in $J^*$. \label{case:p1notbutp}

\item The mini-square $s$ is contained in $U$ and is matched in $J^*$.
\label{case:p1yes}

\end{enumerate}

The values of $\tau^*$ and $\tau$ are the same at most points
and we use ``changing the value'' to indicate the value
is different in the two mappings.

\paragraph{Case \ref{case:p2in}.}  We have that $s$ is now dead in $\rho^*$.
We change one value by restoring the value when $s$ was
introduced to $\pi_2$.

\noindent
{\bf Case \ref{case:p2out}.}  First note that as the pairs
in $J^*$ are disjoint, $s$ remains in $U^*$ for the duration
of the process.  If its partner
from $\pi_2$ is in the super-squares which is the
partner of $S$ in the matching $\pi_1$
then $\tau^*$ equals $\tau$ around $s$.  If it is not then
we change the value at one or two position.
If $s''$ is the partner in $\pi_2$ restore the value of
the variable of the path between $s$ and $s''$ when
$(s,s'')$ was put into $\pi_2$.   If $s$ is in the super-square
of the designated survivor do nothing more and otherwise consider the
super-square, $S'$, that is the partner of $S$
under $\pi_1$.  If there is a node $s' \in S'$ in $U^*$
change the value of one path variable between $s$ and $s'$
from 1 to 0.  We choose one variable that is possible
through our selection function.

\paragraph{Case \ref{case:p2outnot}.}  This is simple as $\tau^*$
equals $\tau$ around $s$.

\paragraph{Case \ref{case:p1not}.}  Also in this case $\tau^*$
equals $\tau$ around $s$.

\paragraph{Case \ref{case:p1notbutp}.}  Suppose $s''$
is the partner in $\pi_1^U$ that disappeared.
We restore the value of the chosen path
between $s''$ and $s$ to 1.  Let $S'$ be
the partner of $S$ under $\pi_1$.  
If there is a node $s' \in S'$ in $U^*$
change the value of one path variable between $s$ and $s'$
from 1 to 0.  We choose one variable that is possible
through our selection function.

\paragraph{Case \ref{case:p1yes}.}
In the case $s$ is now dead in $\rho^*$. We restore
the value on the path used when constructing $\pi_1^U$.

This completes the description of $\tau^*$.  Note that
it might not be exactly an almost complete matching as, due
to the possibility of holes in $U^*$ and the possibility that
the element in the super-square of the designated survivor
is no more the designated survivor.
This causes an unbalance of one at some mini-squares,
but this deviation is completely determined by $U^*$.

The value of $B^*$  plays little role.  We simply set it
to equal $B$.  We might have to change it for 
pairs of mini-squares participating in the above
process, but these are few and this can be handled
by the external information.

We now proceed to define a process that, using external information,
reconstructs the tuple $(\tau, U, \pi_2, B)$ from
$(\tau^*, U^*, \pi_2^*, B^*)$ and the trees $T_i$.
The ``star quadruple'' determines the values of
$y_P$ in the same way as the non-stared
quadruple determined values for $y_P$.  To be more
precise the process is the following.

\begin{itemize}

\item Start with values given by $\tau^*$.

\item Define chosen mini-squares by $U^*$.

\item For any pair of mini-squares in $\pi_1^{U^*}$ 
find one path connecting this pair as
indicated by $B^*$.  Change the corresponding value
from $1$ to $0$.

\item For any pair in $\pi_2^*$ find one
path using $B^*$ and change the value from $1$ to
$0$.

\item Let the live mini-squares be the elements of
$U^*$ and the mini-squares of $\pi_2^*$.

\item Consider a value $y_P$ to be fixed unless if
it is between two live mini-squares.

\end{itemize}

Call this restriction $\rho^*$.  Let us establish
a simple lemma for warm-up just to give an indication
of how the proof proceeds.

\begin{lemma}
The restriction $\rho^*$ forces the first forceable branch
to be followed.  
\end{lemma}

\begin{proof}
We claim that any variable $y_p$ forced by $\rho$ is also forced by $\rho^*$
and to the same value.
This follows as any such variable requires all
its influential mini-squares to be dead.  Any mini-square
that this is dead in $\rho$ is also dead in $\rho^*$.  As the 
situation around it has not changed it is forced to the same
value.

The additional information needed to follow the first forceable branch
is given by the matching pairs in $J_1$.  As this information is included in
$\rho^*$ the first forceable branch is followed.
\end{proof}

To correctly identify the first forceable branch is an important 
step in the reconstruction.  Unfortunately there might be earlier
branches forced to one by $\rho^*$.  The reason is that the
information that makes us follow such a branch might come
from several different $J_i$ and there is no guarantee
that these matchings are locally consistent and hence
it might not constitute an acceptable set $J_1$.  

The optimistic view is that any such branch still would help
us as it must point to some live mini-squares.  We are,
however, not able to use this and we need to
make sure that the reconstruction is not confused.
We introduce the concept of ``signature'' to enable
us to correctly identify the first forceable branch.

\begin{definition}
The {\em signature} of a live mini-square determines whether
it is chosen and in such a case which direction
it is matched in its forceable branch.  It is given by
three bits.
\end{definition}

Let us turn to the reconstruction process. 
During this process we specify the needed external information
and in each situation we give the number of bits needed.
We sum up the total size of the external information
at the end.

We start with $\rho^*$ and we reconstruct
$I_j$ and $J_j$ in order.  We let $\rho_j^*$ be
the restriction obtained from $\rho$ jointly
with $I_i$ for $i < j$ and $J_i$ for $i \geq j$.
In particular, $\rho^*_1$ is simply $\rho^*$ which
is the starting point.
We have a set $E$ of prematurely found chosen mini-squares
jointly with their signatures.  This is initially the
empty set.  We proceed as follows.

\begin{enumerate}

\item Find the next branch forced to one in any
of the $T_i$ by $\rho_j^*$.

\item Find, if any, mini-square of $E$ whose information
is used to follow this path.  If this information is
not locally consistent with $I^{j-1}$, go
to the next branch.

\item Read a bit to determine whether there are more
live variables to be found on the current branch.
In such a case, read a number in $[t]$ to
determine its position.  If this variable has two
influential mini-squares, read another two bits
to determine which of these two mini-squares are alive.
For any alive influential mini-square, retrieve the corresponding 
signature(s) from external information
and, if chosen, include the mini-squares in $E$.
Go to step 2.

\item Reconstruct $I_j$ and $J_j$.  Details below.

\item Remove any chosen mini-square included in $J_j$ from $E$.

\end{enumerate}
We need a lemma.
\begin{lemma}
If the mini-squares of $E$ together with their signatures
on a branch leading to a one is consistent with the information
in $I^{j-1}$, then it is a forceable branch.  In particular the
first such branch after the $j-1$th forceable
branch is the $j$th forceable branch.
\end{lemma}

\begin{proof}
Suppose $v \in E$, which by definition, is a chosen 
mini-square, and was included in $J_{j'}$ for some $j' \geq j$.
As the current path is forced to one it must have been
a potential forceable branch.  If it was not the actual
forceable branch it must be that  at least one
of the matched pairs needed to follow this branch is not allowed.
Forcing information from $\pi_2$ is always allowed and
thus the only possible problem is the consistency on
the chosen mini-squares.  If the signatures of the
mini-squares used on the current branch do not give any conflict 
with $I^{j-1}$,  it was allowed and hence it
must be a forceable branch.
\end{proof}

Whenever we recover a pair of adjacent mini-squares
we use external information to recover the advice
bits.  This only costs 2 bits and thus we below
focus on identifying mini-squares and do not mention
the reconstruction of advice bits.  

We need to discuss how to reconstruct $I_j$ and $J_j$
and start with the latter.  For each matched
pair in $J_j$ we have recovered at least one end-point
from the forceable branch.
If needed we use external information to recover 
the other end-point.  This costs at most $1+\log \Delta$ bits.
This recovers $J_j$ and we need to extend it to $I_j$.

For each element in $J_j$ we need to discover whether it is chosen (this 
is one bit of external information) and to which node it is matched 
in $I_j$.  If it is not chosen then the information in $I_j$
is the same as the one in $J_j$ (and also in $\pi_2$).  
If it is chosen then the partner is either 
in $J_j$ and easy to find or it is live in $\rho^*_j$.
This follows as if it is chosen and belongs to $I_j$ but not $J_j$
then, by Lemma~\ref{lemma:disjoint} it cannot belong to $J^*$.
In the case when it is alive we can, at cost $O(1)+ \log \log n$ bits, find 
its identity.  Once we have identified the
two partners in a matching, at cost $O(1)$ we
can reconstruct which augmenting path was used.  This reconstructs
all of $I_j$.

We reconstruct $I_j$ and $J_j$ and
compute $\rho_{j+1}^*$ and proceed to the next
stage.  At the end of this process we have reconstructed
$\rho$ and in the process we have also identified
all the sets $I_j$ and $J_j$ and we know which pairs
in $J_j$ are chosen and which belong to $\pi_2$.  For the 
pairs of chosen mini-squares, we put them back in to $U$ and
any replacements are moved back to $\pi_2$.  At cost
$O(1)$ we can identify the advice bits of any single
mini-square involved in a move.  This way we recover $U$,
the full $\pi_2$, and $B$.  Once we have these
we can read off the original $\tau$.

Let us calculate how much information
was used.  We have the following contributions.

\begin{itemize}

\item For many mini-square processed we need to
read the signature.  As we only process $O(s)$
mini-squares this is $O(s)$ bits.

\item For the chosen nodes, once we have
established their identities we might need $O(1)$
bits to determine in which direction they are matched.
Note that this can be different in $J_j$ and $I_j$.
This costs at most $O(s)$ bits.

\item For potentially forceable  branches
we read one bit to determine if there is any
additional live mini-square to be found through this branch.
As there are only at most $s$ forceable branches
the answer can be ``no'' at most $s$ times.
As we only discover at most $O(s)$ mini-squares to
put in to $E$, the answer can only by ``yes'' at most $O(s)$ times.
Thus the total number of such bits read is $O(s)$.

\item  For each edge in $J^*$ at least one
end-point is discovered at cost at most $1+\log t$ bits and
the other at cost at most $1+\log \Delta$ bits.

\item For each  mini-square in $I_j$ that need
be discovered since it was not a member of $J_j$
we might need $\log \log n+O(1)$ bits.

\end{itemize}
Summing up the information used we first have
a cost $O(s)$ in several (but constant number) 
places.  The main cost is given by the construction of the
pairs in $J^*$ and this is $s(\log t + \log \Delta)$.
For each mini-square in $J^*$ we might reconstruct
its partner in $I_j$ at an additional total cost 
at most $2s \log \log n$.

We proceed to compare the number of quadruples 
$(\tau, U, \pi_2, B)$ to the number of quadruples
$(\tau^*, U^*, \pi_2^*, B^*)$, and let us first assume that there
are no holes.  Both $U$ and 
$U^*$ contain one mini-square from each super-square but
there is a difference that $U$ contains the designated
survivor while $U^*$ might contain a different mini-square
from this super-square.  Thus the number of possibilities
for $U^*$ is a factor $\Delta$ larger than the
number of possibilities for $U$.

Both $\tau$ and $\tau^*$ are essentially almost complete matchings.
We do not allow $\tau$ to be lopsided but this only reduces
the number of possibilities, by Lemma~\ref{lemma:notfull},
by a factor $1+o(1)$.  It is the case that $\tau^*$
is only slightly more general but we only need
an upper bound on the number of possibilities.
On the other hand, if $U^*$ does not contain the designated
survivor the
sum of all adjacent $y_P$ (in $\pm 1$-notation)
around this mini-square is not zero.  
The number of $\tau^*$ with
these fixed sums is, however, smaller by Lemma~\ref{lemma:max0}.
We conclude that the number of $\tau^*$ is bounded by
$(1+o(1))$ times the number of different $\tau$.

There is no difference between $B$ and $B^*$ and
thus the number of alternatives is the same.
We conclude that if we denote the number of different 
$(\tau, U,B)$ by $N$, then the number of triples
$(\tau^*, U^*,B^*)$ is bounded by $\Delta N (1+o(1))$.

The big difference is in the probability of picking the restriction
 $\pi_2$ and $\pi_2^*$.  The reason is that $\pi_2$ contains $k$
pairs and $\pi_2^*$ only $k-s$ pairs.
It is true that $\pi_2$ is also not unbalanced,
but as this, by Lemma~\ref{lemma:unbalance},
only changes the number by a factor $(1+o(1))$ we ignore
this here and absorb this factor in the error term.   We do have
that both $\pi_2$ and $\pi_2^*$
are matchings and thus we have to be slightly
careful as the number of ways to extend a matching
by one more pair depends on the current matching.
In order to do this we put a probability distribution
also on $\pi_2^*$.   The random process to pick it is
as the one picking $\pi_2$ (described on page \pageref{distribution})
but only with fewer pairs picked and we ignore the balance
conditions.

Let $D$ denote the total number of
possible pairs of mini-squares.  We have that $D= \Delta^2 (n/T)^2 (2+o(1))$.
This follow as there are $\Delta (n/T)^2$ mini-squares. Most of
them (except at the perimeter) have $4 \Delta$ possible partners
and hence the number of pairs is $(1+o(1))\Delta (n/T)^2 4 \Delta/2$.
We have the following lemma.

\begin{lemma}\label{lemma:pi2}
Suppose $\pi_2$ is of size $k$ and $\pi_2^* \subset \pi_2$
is of size $k-s$ and assume that $\Delta= \omega (\log n)$.
Then $Pr [\pi_2] \leq 2^{o(s)}(k/D)^s Pr[\pi_2^*]$
which in its turn is bounded by $2^{O(s)}(\log n/\Delta^2)^s Pr[\pi_2^*]$.
\end{lemma}

\begin{proof}
The way we define the probability space of partial matchings
is slightly unusual and we need to be careful.
There are $k!$ different ways to pick $\pi_2$.
The probability of picking $\pi_2$ in a particular
order $\gamma$ is $\prod_{i=1}^k t_{i,\gamma}^{-1}$ where $t_{i,\gamma}$ 
is the number of pairs of mini-squares available at stage $i$
assuming the order $\gamma$.  As $\Delta = \omega (\log n)$
we have that $k$ is $o( \Delta (n/T)^2)$.  In other words very few of the
mini-squares are picked at any stage of the process
and hence $t_{i,\gamma}$ is on the
form $(1-o(1))D$ but we mostly use much more detailed bounds below.

Similarly, the probability of $\pi_2^*$ can be written as
$$
\sum_{\gamma^*} \prod_{i^*=1}^{k-s} t_{i^*, \gamma^*}^{-1},
$$
where $\gamma^*$ determines the order in which its $k-s$
pairs are added.  
Each picked pair of mini-squares eliminates at most $8 \Delta$ possible
other pairs and hence each $t_{i,\gamma}$ and $t_{i^*,\gamma^*}$ is
at least  $D-8k\Delta$.
The product giving the probabilities do have
many factors so we need more precise information.

Say that $\gamma$ and $\gamma^*$ are associated (written
as $\gamma \sim \gamma^*$) if the $k-s$ common elements
are in the same order.  In such
a situation for each $i^*$ there is an $i$ corresponding
to the same element such that $t_{i^*,\gamma^*} \geq t_{i,\gamma}
 \geq t_{i^*, \gamma^*} -8s\Delta$.
Restricting the product to these paired
indices and using $t_{i,\gamma} \geq D/2$ we have
$$
\prod_{i} t_{i,\gamma} \geq (1-16s \Delta /D)^k \prod_{i^*} t_{i^*,\gamma^*}.
$$
As $k = \Theta( \log n D /\Delta^2)$ we have that
$$(1-16s \Delta /D)^k = 2^{-\Theta ( (s\log n)/\Delta)} = 2^{-o(s)},$$
given the assumption on $\Delta$.
As each factor omitted is at least $D-8k \Delta$ we 
can conclude that 
$$
\prod_{i} t_{i,\gamma} \geq D^s 2^{-o(s)} \prod_{i^*} t_{i^*,\gamma^*}.
$$
There are $k! /(k-s)! \leq k^s $  different $\gamma$ associated
with each $\gamma^*$.  We conclude that
\begin{eqnarray*}
Pr[\pi_2] & =& \sum_{\gamma} \prod_{i=1}^k t_{i,\gamma}^{-1}
= \sum_{\gamma^*} \sum_{\gamma \sim \gamma^*}
\prod_{i=1}^k t_{i,\gamma}^{-1} \leq 
\sum_{\gamma^*} \sum_{\gamma \sim \gamma^*}
D^{-s}  2^{o(s)} \prod_{i^*=1}^{k-s} (t^*)_{i^*,\gamma^*}^{-1} \leq \\
&\leq & \sum_{\gamma^*} 
(k/D)^{s}  2^{o(s)}
\prod_{i=1}^{k-s} (t)_{i^*,\gamma^*}^{-1} \leq 
(k/D)^{s}  2^{o(s)} Pr[\pi_2^*],
\end{eqnarray*}
giving the first bound of the lemma.  The second bound is obtained
using the definitions of $k$ and $D$.
\end{proof}

Let us finish the proof of Lemma~\ref{lemma:switch}.
We can produce a random
tuple $(U, \tau, B, \pi_2)$ that forces a branch
of length at least $s$ in the following way.

\begin{itemize}

\item Pick a random $(U^*, \tau^*, B^*, \pi^*_2)$.

\item Pick external information and use the above
inversion process to obtain $(U, \tau, B, \pi_2)$.

\end{itemize}

Up to a factor $(1+o(1))$ we have $\Delta N$
possibilities for $(U^*, \tau^*, B^*)$ and
$N$ possibilities for $(U, \tau, B)$ where we have
the uniform distribution on both spaces.
It follows from Lemma~\ref{lemma:pi2}
that the probability of the produced four-tuple
is at most $\Delta 2^{O(s)} (\log n/\Delta^2)^s$ that of
initial stared four-tuple.

Recalling that we have $2^{O(s)} (\Delta t)^s (\log n)^{2s}$
possibilities for the external information, 
we can conclude that if we sum
$\Pr[(U, \tau, B, \pi_2)]$ over all tuples producing
a canonical decision tree of depth at least $s$,
then this is bounded by 
$$
2^{O(s)} (\Delta t)^s (\log n)^{2s} \Delta (\log n/\Delta^2)^s
\leq  \Delta 2^{O(s)} ((\log n)^3 t/\Delta)^s.
$$
This proves Lemma~\ref{lemma:switch} in the case when there are no holes.

If we allow $r$ holes then the number of possibilities
of $U^*$ increases by a factor at most $n^{2r}$.
The number of $\tau^*$ is different as some sums are
no longer $0$, but the sum at each mini-square is
uniquely determined by $U^*$.  Once this is fixed, the 
number of $\tau^*$ is at most the number of almost
complete matchings by Lemma~\ref{lemma:max0}.

For there to be a hole, by the balance condition,
all the at least $C \log n /8$ edges between two super-squares must
be present in $J^*$.  This implies that $r= O(s/ \log n)$
and the factor $n^{O(r)}$ can be absorbed in the 
factor $2^{O(s)}$.  

Another difference is that if we have $r$ holes then
$\pi_2^*$ contains $k+r/2-s$ pairs.  Also the factor
resulting from this change can be absorbed in
the $2^{O(s)}$ factor.  This completes the proof also in
case when there are holes.
\end{proof}

\section{The lower bound for proof size}\label{sec:size}

In this section we establish one of the two main theorems
of the paper.

\begin{theorem}\label{thm:size}
Assume that $d \leq O(\frac{\log n}{\log {\log n}})$ and let
$n$ be an odd integer. Then any depth-$d$ Frege proof of
the functional onto $\PHP$ on the $n \times n$ grid requires total size
$$
\exp (\Omega (n^{1/(2d-1)} (\log n)^{O(1)})).
$$
\end{theorem}

\begin{proof}
The proof follows the standard path.  We use
a sequence of restrictions and after the $i$th
restriction any sub-formula of the proof
of original depth at most $i$ is in the range of the $t$-evaluation.

Assuming this, after $d$ restrictions, any sub-formula
of the proof is in the range of the $t$-evaluation.
By Lemma~\ref{lemma:reduced} what remains is a smaller
functional onto $\PHP$ instance and
by Lemma~\ref{lemma:noproof}, provided the
size of the remaining grid is significantly larger than
$t$, the proof cannot derive contradiction.  Let us give
some details.

In the base case of depth $0$, each literal naturally gives
a $1$-evaluation.  Now
suppose that the size (and hence the number of sub-formulas)
of the proof is $2^S$.

We apply a sequence of restrictions, the first
one with $\Delta= \Omega ((\log n) ^3)$ and later with
$\Delta= \Omega (S(\log n) ^3)$ where we assume that
the implied constant is large enough.  We claim that,
with high probability, after step $i$ each sub-formula that was originally
of depth at most $i$ is now in the range of a $2S$-evaluation.

Indeed, consider Lemma~\ref{lemma:switch} and take
constants large enough so that the failure probability for
an individual formula is bounded by $2^{-2S}$.  By the
union bound, except with probability $2^{-S}$, at step
$i$ all formulas originally of depth $i$ are
in the range of a $2S$ evaluation.

The parameter $n$ turns into $\Theta (n/R\Delta^2)$ by one
application of a single restriction and
hence after the $d$ restrictions, and using
that $R=\Theta (\log n)$,  we are left with
a grid of size
$$
n 2^{-cd} S^{2-2d} (\log n)^{-7d}
$$
for some constant $c$.  There is 
a $2S$-evaluation which have every sub-formula
of the proof in its range.  By Lemma~\ref{lemma:noproof}
provided that 
$$
n \geq 300 \cdot 2^{cd} (\log n)^{7d} S^{2d-1}.
$$
such a proof cannot refute the $\PHP$ and we get
a contradiction.  Rearranging we get the bound for $S$
claimed in the theorem.
This concludes the proof.
\end{proof}

Next we turn to the multi-switching lemma.

\section{Multi-switching} \label{sec:mswitch}

The extension to multi-switching only
requires minor modifications and we start by
stating the important lemma.

\begin{lemma}   \label{lemma:multiswitch}
  There is a constant, $A$, such that the
  following holds.
  Consider formulas $F_i^m$, for $m \in [M]$
  and $i \in [n_m]$, each associated with a decision tree of
  depth at most $t$ and let $F^m= \lor_{i=1}^{n_m} F_i^m$.
  Let $\sigma$ be a random full restriction from the 
  space of restrictions defined in Section~\ref{sec:restrictions}.
  Then the probability that $(F^m)_{m=1}^M$
  cannot be represented by an $\ell$ common partial
  decision tree of depth at most $4s$ is at most
  $$
  \Delta M^{4s/\ell} \big(A(\log n)^{6} t \Delta^{-1}\big)^{s}.$$
\end{lemma}

\begin{proof}
As in the proof of the single switching lemma 
we outline a process to construct an $\ell$ common decision tree.
We prove that this process is unlikely to ask more than
$4s$ variables in a similar way to the single switching
lemma.

Let us define the process of constructing an $\ell$ common
decision tree.  
We treat the formulas for increasing values of $m$.
We have counter $j$, starting at $0$, indicating the number
of times we have found a long branch in a decision tree.  
We have an information set $I^+$ which initially is empty.

If a formula $F^m$ admits
a decision tree of depth at most $\ell$ under $\rho$ and
the current $I^+$ then we proceed
to the next formula.  
If not, we set $m_j=m$ and
execute the process of forming a canonical
decision tree.  It discovers a branch of length at least
$\ell$ in this tree and constructs
the corresponding sets $J_i^j$ for $i=1, 2 \ldots g_j$.
We also have the corresponding information sets $I_i^j$.  

After this branch has been discovered we ask, now
in the $\ell$ common decision tree, for the
partner of any chosen element in $I^{j*}= \cup_{i=1}^{g_j} I_i^j$.
The answers jointly with the matching pairs from $\pi_2$ in 
$I^{j*}$ form a new information set $I^{j*+}$.
This information set is added to $I^+$ and we
look for the next formula that needs a decision
tree of depth at least $\ell$.  It might be the
same formula $F^{m_j}$, but the situation is new as
we have an updated set $I^+$.

Properties follow along the same lines as in the
single switching lemma and we start with the lemma
corresponding to Lemma~\ref{lemma:disjoint}.

\begin{lemma}\label{lemma:mdisjoint}
The support sets of $J_i^j$ are disjoint.  The support of
$J_i^j$ is also disjoint with the support of $I_{i'}^{j'}$ as
long as $i'\not= i$ or $j'\not= j$.
\end{lemma}

\begin{proof}
The parts from $\pi_2$ are obviously disjoint and
we need only consider the chosen mini-squares.
For the same value of $j$, this lemma follows by the
same argument as Lemma~\ref{lemma:disjoint}.  As $I^{j*+}$
contains each $J^{j}_i$ and sets $J_i^{j'}$ for $j'> j$ are
disjoint from $I^{j*+}$ the general case is
also true.
\end{proof}

Let $J^{j*}=\cup_{i} J_i^j$ be the information set obtained
during the $j$th stage.
It contains at least $\ell$ matched
pairs and let us assume that the true number is $s_j$. 
Let  $J^*=\cup_j J^{j*}$ and define $\rho^*$ as
$\rho$ joint with information set $J^*$.  
We need to describe a process to find $\rho$
from $\rho^*$ and the set of formulas.

Let us
immediately point out that the information which $F^{m_j}$
are processed is contained in the external information.
As at most $4s/\ell$ formulas are processed this gives an
extra factor $M^{4s / \ell}$ in the size of the external
information.  This factor appears in the bounds of the
lemma and hence we proceed to analyze the rest of the
process.

Once we know the value of $m_j$ we can run the
reconstruction process based on  $F^{m_j}$ finding the
sets $I_i^j$ and $J_i^j$ exactly
as in the single switching lemma. There is no need to change
this procedure.

A difference compared to the single switching lemma is
the sets $I^{*j+}$.  Again the information pieces from
$\pi_2$ do not give a problem as they never change.
For each chosen mini-square in $I_i^j$ we need to specify a partner 
inside $I^{*j+}$.  This may or may not belong to some $I_{i'}^j$.
If it does, we can specify this at cost $O(1)$ but if not
we can specify it at cost $O(1)+ \log \log n$ bits as it is alive
in $\rho^*$.

For each information piece in $J_i^j$ we thus have,
in the worst case, to discover one mini-square
at cost $t$, one at cost $\Delta$ and then
six mini-squares at cost $O(1)+\log \log n$ bits.  These
are two partners in $I_i^j$ and up to four
partners in $I^{*j+}$.
This implies that the total external cost at stage
$j$ is at most $s_j(O(1)+6 \log \log n+\log t + \log {\Delta})$ bits.

As any edge in $J_i^j$ might result in four queries in the
common decision tree, if we have $4s$ questions in
this tree, then $\sum s_j\geq s$ and let us for convenience
assume we have equality.  The total cost of external information
is in such a case $2^{O(s)} (\log n)^{6s} (t \Delta)^s$.

Lemma~\ref{lemma:pi2} applies also in the current case and from 
a similar calculation to that of the standard
switching lemma, we conclude that the total probability
of quadruples $\tau$, $U$ $\pi_2$, $B$
that can be reconstructed is bounded by 
$$M^{4s/\ell}\Delta 2^{O(s)} ((\log n)^{6} t \Delta^{-1})^s.$$
This completes the proof.
\end{proof}

\section{The lower bound for number of
lines}\label{sec:nlines}

We finally arrive at our second main theorem.

\begin{theorem}\label{thm:lines}
Assume that $d \leq O(\frac{\log n}{\log {\log n}})$ and let
$n$ be an odd integer and $M$ a parameter. 
Then any depth-$d$ Frege proof of
the functional and onto $\PHP$ on the $n \times n$ grid where each line
is of size $M$ requires at least
$$
\mbox{exp} \left(\Omega \left(\frac {n}{(\log M)^{2d-1} (\log n)^{O(1)}}\right)\right)
$$
lines.
\end{theorem}

\begin{proof}
Given that we established the multi-switching lemma
the situation is now essentially the same as 
in \cite{jhkr}.  

Parameters are similar as in the proof for the total proof
size.   We use $\ell= \log M$ in all applications
of the Lemma~\ref{lemma:multiswitch}
while we have $t=1$ in the
first application and in later applications $t=\ell$.
With these choices $M^{4s /\ell}$ equals
$2^{O(s)}$ and can be included in the constant $A$.

We set $\Delta = \Theta ((\log n)^7 t)$ to ensure that the 
base of the exponent (including the contribution from the $M^{4s/ \ell}$
factor) is bounded by $\frac 1{32}$.  Now suppose the
proof has $N$ lines.  

In our first application we set $4s_1= \log N$ and
we conclude, by the union bound,  that each line with
high probability admits its own $s_1$-common
$\ell$-evaluation.

In the second iteration we have $2^{s_1} N $ sets of
formulas to consider.  This is the case as we need to consider
each leaf of each decision tree in each line.
By setting $s_2= 2 s_1$ we can apply the union
bound and conclude that with high probability all
depth 2 formulas are now in the range of the
respective $\ell$-evaluations.

Continuing similarly, in the $i$th iteration we use $s_i=2^{i-1} \log N$.
The final common decision tree is 
of depth at most $\sum_{j=0}^{i-1} s_j= O(2^d \log N)$.

After $d$ rounds of restrictions we have reduced the 
side length of the grid on which we define the $\PHP$ from $n$ to $n/ (\log n)^{O(d)} (\log M)^{2d-2}$.
If this is larger than $300(\sum_{j=0}^{d} s_j + \log M)$
we can conclude that there is no proof and this gives
the bound of the theorem.
\end{proof}

\section{Final words}\label{sec:conclusion}

It seems that to be able to prove a switching lemma for a
space of restrictions one essential property is that for any
given variable, the probability of setting it to either constant
should be significantly higher than keeping it undetermined.
In the unrestricted $\PHP$ the probability that any variable is true is 
about $\frac 1n$ and to make the probability of the variables
remaining undetermined smaller we must go from size $n$ to a size
smaller than $\sqrt n$.  With such a quick decrease in
size we can only apply the switching lemma $O (\log \log n)$
times.  In the graph $\PHP$ on the grid the
probability of a variable being true to about $\frac 14$
and at the same time the probability of keeping a variable
undetermined is about $\log n/\Delta$.
It is harder to preserve a graph $\PHP$, but this is 
made possible by using augmenting paths as the new
variables.  Thus it seems that both decisions were forced
upon us but certainly there might be other possibilities and
it would be interesting to see alternative proofs for a
switching lemma for a space of restrictions that preserve
the $\PHP$.

In this paper we have proved yet another switching lemma and 
it might not even be the last one.  There are many remaining
questions in both proof complexity and circuit complexity 
and some might be attackable by these types of techniques.
It would be very interesting, however, to find a different
technique to attack questions of Frege proofs where each
formula is relatively simple.

While proofs of switching lemmas are non-trivial, the 
properties of the proof we use are rather limited.
In the assumed proof of contradiction, after the
restrictions, the proof, more or less, only contain formulas of constant
size.  It is not difficult to see that such
a proof cannot find a contradiction for 
a set of axioms that are locally consistent.  
It would be interesting with a reasoning that used
more interesting properties of being a proof.

\section*{Acknowledgments}  

I am deeply grateful to
Svante Janson for suggesting the proof of Lemma~\ref{lemma:notfull}.
I am also very grateful to Per Austrin, Kilian Risse, Ben Rossman,
and Aleksa Stankovi\'c
for a sequence of discussions that lead to the start of the
ideas for the current paper.  
Kilian has also greatly helped to improve the quality of the writing
of the current paper.  This is partly due to discussions about
writing \cite{jhkr}, but also when discussing details of
the current paper.
Finally, I thank the referees of
this paper for several suggestions on how to improve the 
presentation.
The research leading to this publication
was supported by the Knut
and Alice Wallenberg foundation.  

\printbibliography

@article{agg01,
author={Albert Atserias and Nicola Galesi and Ricard Gavaldà},
title={Monotone Proofs of the Pigeon Hole Principle},
journal={Mathematical Logic Quaterly},
Volume={47},
Year={2001},
pages={461-474},
doi={10.1007/3-540-45022-X_13}
}

@article{Ajtai94Complexity,
  title   = {The Complexity of the Pigeonhole Principle},
  author  = {Ajtai, Miklos},
  journal = {Combinatorica},
  volume  = {14},
  number  = {4},
  pages   = {417\nobreakdash--433},
  year    = {1994},
  doi={10.1007/BF01302964}
}

@article{BPU92ApproximationAS,
  title={Approximation and Small-Depth Frege Proofs},
  author={Stephen Bellantoni and Toniann Pitassi and Alasdair Urquhart},
  journal={SIAM J. Comput.},
  year={1992},
  volume={21},
  pages={1161-1179},
  doi={10.1137/0221068}
}

@book{bollobas,
  title={Combinatorics: Set Systems, Hypergraphs, Families of Vectors, and Combinatorial Probability},
  author={B{\'e}la Bollob{\'a}s},
  isbn={9780521337038},
  lccn={lc86009602},
  url={https://books.google.se/books?id=psqFNlngZDcC},
  year={1986},
  publisher={Cambridge University Press}
}

@article{fournier,
author={Jean-Claude Fournier},
title={Tiling pictures of the plane with dominoes},
journal={Discrete Mathematics},
volume={165/166},
year={1997},
pages={313-320},
doi={10.1016/S0012-365X(96)00179-3}
}

@article{fss,
author ={Merrick L. Furst and
                  James B. Saxe and
                  Michael Sipser},
title = "Parity, circuits and the polynomial-time hierarchy",
journal = {Mathematical Systems Theory},
Volume = {17},
pages ={13-27},
year = {1984},
doi={10.1007/BF01744431}
}

@inproceedings{jhswitch,
 author = {Johan H{\aa{}}stad},
 title = {Almost optimal lower bounds for small depth circuits},
 booktitle = {Proceedings of the eighteenth annual ACM symposium on Theory of computing},
 series = {STOC '86},
 year = {1986},
 location = {Berkeley, California, United States},
 pages = {6--20},
 numpages = {15},
 publisher = {ACM},
 address = {New York, NY, USA},
 doi={10.1145/12130.12132}
}

@inproceedings{jhphpfocs,
author= {Johan H{\aa}stad},
title={On small-depth {Frege} proofs for \protect{PHP}},
booktitle={Proceedings, 64th Annual IEEE Symposium on Foundations of
Computer Science},
year={2023},
series={FOCS 2023},
pages={37-49},
publisher={IEEE},
doi={10.1109/FOCS57990.2023.00010},
location = {Berkeley, California, United States}
 }

@article{jhmultiswitch,
author={Johan H{\aa{}}stad},
journal={SIAM Journal on Computing},
title={On the correlation of parity and 
small-depth circuits},
year={2014},
volume={43},
pages={1699-1708},
doi={10.1137/120897432}
}

@ARTICLE{jhtseitin,
author={Johan H{\aa{}}stad},
journal={Journal of the ACM},
title={On Small-Depth {Frege} Proofs for {Tseitin} for Grids},
year={2020},
volume={68},
number={},
pages={1-31},
doi={10.1145/3425606}
}

@article{jhkr,
author={Johan H{\aa{}}stad and Kilian Risse},
  journal={SIAM Journal of Computation},

  title={On Bounded Depth Proofs for Tseitin Formulas on the Grid; Revisited}, 

  year={2025},

  volume={54},

  pages={FOCS22-288-FOCS22-339},
  doi={10.1137/22M153851X}

}

@Article{haken,
  Title                    = {The intractability of resolution},
  Author                   = {Armin Haken},
  Journal                  = {Theoretical Computer Science},
  Year                     = {1985},
  Pages                    = {297 - 308},
  Volume                   = {39},
  doi={10.1016/0304-3975(85)90144-6}
}

@article{KPW95FregePHPExp,
author={Krajíček, Jan and Pudlák, Pavel and Woods, Alan},
title = {An exponential lower bound to the size of bounded depth {Frege} proofs of the pigeonhole principle},
journal = {Random Structures \& Algorithms},
volume = {7},
number = {1},
pages = {15-39},
doi = {10.1002/rsa.3240070103},
url = {https://onlinelibrary.wiley.com/doi/abs/10.1002/rsa.3240070103},
eprint = {https://onlinelibrary.wiley.com/doi/pdf/10.1002/rsa.3240070103},
year = {1995}
}

@inproceedings{imp,
author={Russell Impagliazzo and William Matthews and Ramamohan Paturi},
title={A satisfiability algorithm for \protect{$AC^0$}},
booktitle={Proceeding of the 23rd Annual ACM-SIAM Symposium on Discrete
Algorithms, SODA 2012},
year={2012},
pages={961-972},
doi={10.1137/1.9781611973099.77}
}

@article{PBI93ExponentialLowerBounds,
  author    = {Toniann Pitassi and Paul Beame and Russell Impagliazzo},
  title     = {Exponential Lower Bounds for the Pigeonhole Principle},
  journal   = {Computational Complexity},
  volume    = {3},
  year      = {1993},
  pages     = {97\nobreakdash--140},
  doi       ={10.1007/BF01200117}
  }

@inproceedings{pitassi16frege,
author = {Pitassi, Toniann and Rossman, Benjamin and Servedio, Rocco A. and Tan, Li-Yang},
title = {Poly-Logarithmic Frege Depth Lower Bounds via an Expander Switching Lemma},
year = {2016},
isbn = {9781450341325},
publisher = {Association for Computing Machinery},
address = {New York, NY, USA},
url = {https://doi.org/10.1145/2897518.2897637},
doi = {10.1145/2897518.2897637},
booktitle = {Proceedings of the Forty-Eighth Annual ACM Symposium on Theory of Computing},
pages = {644–657},
numpages = {14},
location = {Cambridge, MA, USA},
series = {STOC 2016}
}

@INPROCEEDINGS {PRT21,
author = {Toniann Pitassi and Prasanna Ramakrishnan and Li-Yang Tan},
booktitle = {62nd Annual Symposium on Foundations of Computer Science, FOCS 2021},
title = {Tradeoffs for small-depth Frege proofs},
year = {2021},
volume = {},
issn = {},
pages = {445-456},
doi = {10.1109/FOCS52979.2021.00052},
url = {https://doi.ieeecomputersociety.org/10.1109/FOCS52979.2021.00052},
publisher = {IEEE Computer Society},
address = {Los Alamitos, CA, USA},
month = {feb}
}

@InBook{Razborov1995,
  Title                    = {Bounded Arithmetic and Lower Bounds in Boolean Complexity},
  Author                   = {Alexander A. Razborov},
  Note       = {Editors Peter Clote and Jeffrey Remmel},
  Pages                    = {344--386},
  Publisher                = {Birkh{\"a}user Boston},
  Year                     = {1995},

  Address                  = {Boston, MA},

  Booktitle                = {Feasible Mathematics II},
  ISBN                     = {978-1-4612-2566-9},
}

@inproceedings{rst,
  author    = {Benjamin Rossman and
               Rocco A. Servedio and
               Li{-}Yang Tan},
  title     = {An Average-Case Depth Hierarchy Theorem for Boolean Circuits},
  booktitle = {{IEEE} 56th Annual Symposium on Foundations of Computer Science, {FOCS}
               2015, Berkeley, CA, USA, 17-20 October, 2015},
  pages     = {1030--1048},
  year      = {2015},
  doi= {10.48550/arXiv.1504.03398}
}

@article{thiant,
author= {Nicolas Thiant},
title={An \protect{$O(n \log n)$} algorithm for finding a domino tiling of a plane
picture whose number of holes is bounded},
journal={Theoretical Computer Science},
volume={303},
year={2003},
pages={353-374},
doi={10.1016/S0304-3975(02)00497-8}
}

@inproceedings{tseitin,
author={Grigori~S. Tseitin},
title={On the complexity of derivation in the proposistional calculus},
editor={A. O. Slisenko},
booktitle={Studies in constructive mathematics and mathematical logic, Part II},
year={1968},
doi={10.1007/978-3-642-81955-1_28}
}

@inproceedings{sipser,
 author = {Michael Sipser},
 title = {Borel sets and circuit complexity},
 booktitle = {Proceedings of the fifteenth annual ACM symposium on Theory of computing},
 series = {STOC 1983},
 year = {1983},
 isbn = {0-89791-099-0},
 pages = {61--69},
 numpages = {9},
 publisher = {ACM},
 address = {New York, NY, USA},
 doi={10.1145/800061.808733}
}

@article{uf,
  author    = {Alasdair Urquhart and Xudong Fu},
  title     = {Simplified Lower Bounds for Propositional Proofs},
  journal   = {Notre Dame Journal of Formal Logic},
  volume    = {37},
  number    = {4},
  year      = {1996},
  pages     = {523\nobreakdash--544},
  doi       = {10.1305/ndjfl/1040046140}
}

@INPROCEEDINGS{yao,
author={Andrew Chi-Chih Yao},
booktitle={26th Annual Symposium on Foundations of Computer Science, FOCS 1985}, title={Separating the polynomial-time hierarchy by oracles},
year={1985},
month={oct.},
volume={},
number={},
pages={1 -10},
doi={10.1109/SFCS.1985.49},
ISSN={0272-5428},}

\newpage
\appendix

\section{Proof of Lemma~\ref{lemma:conssub}} \label{sec:proof}

The purpose of this section is to prove Lemma~\ref{lemma:conssub}, restated
here for convenience.

\Conslemma*

\begin{proof}
Suppose the matching $M$ is of size $t$ and
let $M_0$ be a sub-matching of size $t_0$.  We have
sets of $S$ of $T$ and an extension
of $M$ to a matching of $S \times T$.  This,
of course, is also an extension of $M_0$ and if both these
sets are of size at most $48t_0$ we are done,
so assume that this is not the case and we need
to decrease the sizes of $S$ and $T$.

We need to take a
closer look on which sub-areas of the grid allow
matchings.  This is traditionally stated
as ``tiling an area in the plane with dominos of size 2''.
We need to recall some facts about such tilings
and we use results from \cite{thiant} and \cite{fournier}.
For the convenience of the reader we give the argument
from scratch.   

A ``figure'' is a number of black
and white squares.  An example is given in Figure~\ref{fig:ex1}.
Each square has four edges
along its sides and two squares are connected 
if they share an edge (and not if they only
meet at a corner).  Connected squares are thus of
different colors.  We are interested in connected
figures but not necessarily simply connected.
Each figure has a perimeter which is the set of edges
which only belong to one square of the figure.
There is an outside perimeter which is a simple path and
if the figure is not simply connected, we have
perimeter(s) around any interior hole(s).

\medskip
\begin{figure}[h]
\begin{center}
\begin{tikzpicture}[scale=1.288]

\fill[lightgray] (1.0,0) rectangle(1.5,.5);
\fill[lightgray] (0.5,0.5) rectangle(1.0,1.0);
\fill[lightgray] (1.5,0.5) rectangle(2.0,1.0);
\fill[lightgray] (2.5,0.5) rectangle(3.0,1.0);
\fill[lightgray] (2.0,1.0) rectangle(2.5,1.5);
\fill[lightgray] (0.5,1.5) rectangle(1.0,2.0);
\fill[lightgray] (1.5,1.5) rectangle(2.0,2.0);
\fill[lightgray] (1.0,2.0) rectangle(1.5,2.5);
\fill[lightgray] (2.0,2.0) rectangle(2.5,2.5);

\fill[black] (0.5,0) rectangle(1,.5);
\fill[black] (1.5,0) rectangle(2,.5);
\fill[black] (0,.5) rectangle(0.5,1);
\fill[black] (1,.5) rectangle(1.5,1);
\fill[black] (2,.5) rectangle(2.5,1);
\fill[black] (0.5,1) rectangle(1,1.5);
\fill[black] (0,1.5) rectangle(0.5,2);
\fill[black] (1,1.5) rectangle(1.5,2);
\fill[black] (2,1.5) rectangle(2.5,2);

\draw (1.0,0) rectangle(1.5,.5);
\draw (0.5,0.5) rectangle(1.0,1.0);
\draw (1.5,0.5) rectangle(2.0,1.0);
\draw (2.5,0.5) rectangle(3.0,1.0);
\draw (2.0,1.0) rectangle(2.5,1.5);
\draw (0.5,1.5) rectangle(1.0,2.0);
\draw (1.5,1.5) rectangle(2.0,2.0);
\draw (1.0,2.0) rectangle(1.5,2.5);
\draw (2.0,2.0) rectangle(2.5,2.5);

\draw (0.5,0) rectangle(1,.5);
\draw (1.5,0) rectangle(2,.5);
\draw (0,.5) rectangle(0.5,1);
\draw (1,.5) rectangle(1.5,1);
\draw (2,.5) rectangle(2.5,1);
\draw (0.5,1) rectangle(1,1.5);
\draw (0,1.5) rectangle(0.5,2);
\draw (1,1.5) rectangle(1.5,2);
\draw (2,1.5) rectangle(2.5,2);

\draw [->] (2.0,0) -- (1.7,0);
\draw [->] (1.5,0) -- (1.2,0);
\draw [->] (1.0,0) -- (0.7,0);
\draw [->] (0.5,0) -- (0.5,0.3);
\draw [->] (0.5,0.5) -- (0.2,0.5);
\draw [->] (0,0.5) -- (0,0.8);
\draw [->] (0,1.0) -- (0.3,1.0);
\draw [->] (0.5,1.0) -- (0.5,1.3);
\draw [->] (0.5,1.5) -- (0.2,1.5);
\draw [->] (0,1.5) -- (0,1.8);
\draw [->] (0,2) -- (0.3,2);
\draw [->] (0.5,2) -- (0.8,2);
\draw [->] (1,2) -- (1.0,2.3);
\draw [->] (1,2.5) -- (1.3,2.5);
\draw [->] (1.5,2.5) -- (1.5,2.2);
\draw [->] (1.5,2) -- (1.8,2);
\draw [->] (2,2) -- (2,2.3);
\draw [->] (2,2.5) -- (2.3,2.5);
\draw [->] (2.5,2.5) -- (2.5,2.2);
\draw [->] (2.5,2) -- (2.5,1.7);
\draw [->] (2.5,1.5) -- (2.5,1.2);
\draw [->] (2.5,1) -- (2.8,1);
\draw [->] (3,1) -- (3,0.7);
\draw [->] (3,0.5) -- (2.7,0.5);
\draw [->] (2.5,0.5) -- (2.2,0.5);
\draw [->] (2,0.5) -- (2.0,0);

\draw [->] (2.0,1) -- (2,1.3);
\draw [->] (2,1.5) -- (1.7,1.5);
\draw [->] (1.5,1.5) -- (1.2,1.5);
\draw [->] (1.0,1.5) -- (1.0,1.2);
\draw [->] (1.0,1.0) -- (1.3,1.0);
\draw [->] (1.5,1.0) -- (1.8,1);

\end{tikzpicture}
\end{center}
\caption{A figure that cannot be tiled.  It is not simply
connected and we show the assumed orientation of the perimeter.} \label{fig:ex1}
\end{figure}

It is convenient to use directed paths and
we direct the perimeter of a figure in the clockwise
order and the perimeters of the holes in the counter-clockwise
order.  In each case, the figure lies to the right of each
perimeter edge when the edge is traversed in its
intended direction.
For notational convenience we let the coordinates of the
corners of the squares be the integer points.

\begin{definition}
The {\em cost} of a directed edge is 1 if the square to its
right (when traversed in the intended direction) is white
and -1 otherwise.  The cost of a path is the sum
of the costs of its edge.  We denote the cost of
a path $P$ by $c(P)$.
\end{definition}

The following lemma is one reason why the given
definitions are useful.

\begin{lemma}\label{lemma:boundary}
Assume we direct the perimeter as described, i.e.
the outer perimeter in the clockwise direction and
the perimeter around any holes in the counter-clockwise
direction.  If the figure contains
$b$ black squares and $w$ white squares
then the total cost of the perimeter is $4(w-b)$.
\end{lemma}

\begin{proof}
For each square in the figure consider
the four edges of the perimeter of this
square when it is traversed in
the clockwise order.   The total cost
of all edges is then $4(w-b)$ as each white square contributes
4 and each black square -4.  The 
edges internal to the figure are traversed once
in each direction and these two contributions
cancel each other.  The edges traversed once
give the cost of the perimeter and hence
the lemma follows.
\end{proof}

Suppose we are given a figure, $A$, with equally many
black and white squares and we want to determine
whether it admits a matching.  As we can treat
each connected component separately we can assume
that $A$ is connected.
We use Hall's marriage theorem to
study this question.  Let $L$ be a set of
black squares and $N(L)$ be the set of white
neighbors, both within $A$.  It follows from 
Hall's marriage theorem that a necessary and sufficient
condition for $A$ to be matchable is that for any $L$ we have $|N(L)| \geq L$.
For a set $L$ that potentially violates this condition let us
look at the figure $A_0=L \cup N(L)$.  This contains
more black squares than white squares and hence
the cost of its perimeter must be negative.

This perimeter contains edges from the perimeter
of $A$ and some edges internal to $A$.  Note that the 
latter edges all have cost one as there must be a white
square to its right.  Indeed if there was a black square
to its right, $N(L)$ would contain also the square to its left
and hence the edge would not belong to the perimeter.
This property implies that these edges internal to
$A$ must alternate in the sense that every other
edges is vertical and every other edge is horizontal.
We call such a sub-figure ``white bounded''.  Using this
terminology we have the following lemma which is a combination
of the above reasoning and Lemma~\ref{lemma:boundary}.

\begin{lemma}\label{lemma:nomatch}
A figure with an equal number of black and white squares
is matchable unless it contains
a connected white-bounded sub-figure with a perimeter
of negative cost.
\end{lemma}

For instance in Figure~\ref{fig:ex1} one such sub-figure
can be found by 10 squares down and to the left.
This sub-figure contains 6 black and 4 white
squares.

Let us return to our matching $M_0$
which can be extended to a matching of $S \times T$
where $S$ and $T$ both are of total size 
at most $48t$.  The aim is to decrease these sizes.
In order to do this we define a new notion.  Let
$I= [a,b]$ be an interval and $K$ a multiset contained
in the interior of $I$.  Define the ``left function'' of
$I$ as follows.

\begin{itemize}

\item $f^I_l(a)=0$.

\item For $a < i \leq b$ set $f^I_l (i) = f^I_l(i-1)+\delta_i$ where
$\delta_i = 2$ if $i \not \in K$ and otherwise $\delta_i = -k_i/2$ where
$k_i$ is the number of copies of $i$ in $K$.

\end{itemize}

Similarly, we have a right function, defined by $f^I_r(b)=0$ and
$f^I_r(i)=f^I_r(i+1)+\delta_i$.   A set, $T$,
covers a set $K$ if for each $x\in K$ we have $x\in T$.  
We need slightly more than a standard cover and to formulate
this it is natural to think of any set of integers as
a set of maximal and disjoint intervals.  It is easy to
see that this can be done in a unique way.  As the correspondence
is immediate we blur the distinction between the set of points
and the collection of maximal and disjoint intervals.

\begin{definition}
A set of intervals, $S$,
{\em well covers} a multi-set $K$
if $S$ covers $K$ and for each $I\in S$ 
we have $f_l^I (i) \geq 0$ and $f_r^I (i) \geq 0$ for
each $i \in I$.
\end{definition}

We have the following lemma. 
\begin{lemma} \label{lemma:wellcover}
Any set $K$ of cardinality $t$ can be well covered by
a collection of intervals of even length of total size
$6t$.
\end{lemma}

\begin{proof}
We add the points of $K$ one by one and build the set of
intervals at the same time maintaining the
property that the current intervals form a well cover of the
so far added points.  For each point in $K$ we add
at most 5 points to $S$.  We make the intervals of even length
in the end.  

We start with $K$ empty and no intervals.
Suppose we want to add $p$ to $K$.
If $p$ is not in $S$, we add it.  We find the
two points to the left of $p$ which are as close as
possible to $p$ and which are currently not in $S$ and add
them to $S$.  Similarly we add two points to the right of $p$.

Obviously $S$ cover $K$ and we need to prove that it also
well covers.  We establish
this by induction.  Clearly this is true if $K$ is a
singleton and we need to study adding a general point $p$ to $K$.

As intervals not touched by the change remain the same,
and left and right are symmetric, it is sufficient to look
at $f^I_l$ for the interval, $I$, that the added point $p$
ends up in.  The interval $I$ is the union of some
intervals $I^i$ existing before the addition of $p$,
jointly with $p$ and the four added points.

Define a function $f$ on $I$ by setting $f(x)=f_{l}^{I^i}(x)$
if $x\in I^i$ and $f(x)=0$ otherwise.  We claim
that $f_l^I(x) \geq f(x)+\delta_x$ where, for $x$ strictly
to the left of $p$, $\delta_x$ is twice
the number of points added strictly to the left $x$.
If $x$ equals $p$ or is to the right of $p$ then $\delta_x = 0$.
Clearly this claim is sufficient to establish that $f_l^I$ remains non-negative.

The claim uses the alternative characterization of $f^I_l(x)$ as being
twice the number of points to the left of $x$ that belongs to  $I\setminus K$
minus half the number of points in $I \cap K$ to the left of or
equal to $x$.
If there was no merger of interval the bound on $f_l^I$ is obvious.
That mergers do not decrease $f_l^I$ follows from
the fact that the contribution to $f_l^I$ from $I^i$ is non-negative
due to the non-negativity of $f_l^{I^i}$.  We conclude that the
claim holds.

Finally, we add a point to any interval in $S$ of odd
size.  As each interval is of length at least 5,
this multiplies its length by at most $6/5$.
\end{proof}

The squares of $M_0$ gives a number of connected figures
and consider any point on the perimeter of one of these
figures.  We define $K_1$ to be
the set of all first coordinates of these points, considered
as a multiset.  This is clearly a set of size at
most $8t_0$.   We let $S_0$ be the
set given by Lemma~\ref{lemma:wellcover} that well
covers $K_1$.  Similarly considering the second
coordinates we get a set $K_2$ and let $T_0$ be a well cover
of this set.   If either $S_0$ or $T_0$ contains points
outside $[n]$ we drop these point(s).  
By Lemma~\ref{lemma:wellcover} both $S_0$ and $T_0$ are
of size at most $48t_0$ and hence Lemma~\ref{lemma:conssub}
follows from Lemma~\ref{lemma:s0t0} below.
\end{proof}

\begin{lemma}\label{lemma:s0t0}
$M_0$ can be extended to a matching 
of $S_0 \times T_0$.
\end{lemma}

\begin{proof}
Define the figure $A$ to be $S_0 \times
T_0$ with the squares of $M_0$ removed.  Suppose for
contradiction that $A$ cannot be tiled with dominos of size 2.
Then by Lemma~\ref{lemma:nomatch} we get a white-bounded
connected sub-figure, $A_0$, with a perimeter 
of negative cost.  As we only analyze $A_0$, we only need to
consider its connected component in $S_0 \times T_0$
and hence we can assume
that $S_0$ and $T_0$ both are given by single intervals.
The perimeter of $A_0$ contains an outer boundary and the
perimeters of its holes (if any).  First we observe that cost of
the perimeter of any hole is non-negative.  Suppose on the contrary
that you have a hole $H$ with negative cost perimeter.

To have any chance to have a negative cost, the hole $H$ must
contain some elements of $M_0$ and these give
a matching $M_0'$.  If not it has an all white boundary
and positive cost.  Now consider any rectangle, $R$ with even
side lengths that strictly contains $H$.  We claim
that there is no extension of $M_0'$ to cover all of $R$.
This follows as the figure obtained by removing from
$H$ from $R$ has a negative cost boundary.
Furthermore, this figure is of the form $L \cup N(L)$ for a suitable
$L$ and its existence would contradict that $M_0'$ can be extended
to any rectangle with even side lengths that contains
its component in $S\times T$. 

This proves that no hole in $A_0$ contributes a negative
cost and next we establish that also the boundary
perimeter of $A_0$ is of non-negative cost.

If this boundary contains no part of the boundary of $S_0 \times T_0$
then we can apply the above argument and conclude that
the cost is positive.  On the other hand if it is identical
to the boundary of $S_0 \times T_0$ then the cost is 0.

In the remaining case we can 
divide this perimeter into segments $s_i$ numbered
form $0$ to $2t-1$ where even numbered segments
are along the perimeter of $S_0 \times T_0$ and
odd segments contain edges that either are in
the interior of $A_0$ or along the perimeter of
$M_0$.  First note that $c(s_{2j})\geq 0$ for any
$j$.  This follows as they start and end with
edges of cost $1$.  Costs of edges along the perimeter
alternate between $1$ and $-1$ except at the corners
where we get two adjacent edges of the same cost.

Now let us look at the odd numbered components.
\begin{lemma}\label{lemma:poscost}
For each odd $i$ we have $c(s_i) \geq  0$.
\end{lemma}

\begin{proof}
Suppose that $s_i$
starts at point $(a,b)$.  Then either $a$ is in the interior
of $S_0$ or $b$ is in the interior of $T_0$ and as the two 
cases are symmetric we may assume the former.
Let us for the time being assume that $b$ is the smallest
element in $T_0$.   If $s_i$
does not contain any edges from the perimeter of a
component of $M_0$, each edge on the segment has
cost $1$ and the lemma is true.  Otherwise, let $p$ be the last
point of $s_i$ that is on the perimeter of $M_0$.   We claim that the
cost of the part of $s_i$ from $(a,b)$ to $p$ is non-negative.
This establishes the lemma as all
edges after $p$ on $s_i$ are in the interior of $A_0$ and
hence of cost $1$.
Let $b'$ be the largest value of any second coordinate 
of any point on
the segment from $(a,b)$ to $p$.  The following
claim, of which we postpone the proof, finishes the proof,
in the case when $b$ is the smallest element of $T_0$.
\begin{claim}\label{claim:lbound}
The path from $(a,b)$ to $p$ has cost at least $f_l^{T_0} (b')$.
\end{claim}
The case when $b$ is the largest element in $T_0$ is completely
analogous.  We define $b'$ to be the lowest value of the
second component of the segment and use $f_r^{T_0}$ 
in the corresponding claim.  This completes the proof
of Lemma~\ref{lemma:poscost} modulo the claim.
\end{proof}

Let us establish Claim~\ref{claim:lbound}.  Look at the projection
of the given path on the second coordinate.  Consider any $b''$
with $b <b'' \leq b'$ such that no point of $M_0$ has second
coordinate $b''$.  We want to associate cost two with any
such value.
As no point in $M_0$ projects on $b''$ any edge with
a second coordinate $b''$ must be in the interior of $A_0$.
As previously observed in this part we alternate
between vertical and horizontal edges.

Thus there exists one edge with both end points with
second coordinate equal to $b''$.  We associate the
cost $1$ to $b''$ due to this edge.  We must also have
two edges with different second coordinates
and one end point with second coordinate $b''$.
In most situation there is one edge with second coordinate
$b''-1$ and one edge with second coordinate $b''+1$.
The only exception is when $b''=b'$ and this points
does not belong to $M_0$ in which case both other
end points have second coordinate $b''-1$.
For each of these two edges we associate cost $1/2$ with
$b''$ for a total cost of two associated with   $b''$.
Note also that the total cost caused by any edge
in the interior is $1$.

The only edges on the path from $(a,b)$ to $p$ that
give value $-1$ are those from the perimeter of $M_0$.  
We associate the cost $-1/2$ with each of the second
coordinates of its end-points.  These associated costs follow
closely the definition
of $f_l^{T_0}$ and the claim follows
by simple inspection.  Having established the claim we conclude
that the perimeter of the figure $A_0$ cannot have negative
costs.  This contradiction completes the proof of
Lemma~\ref{lemma:s0t0}.
\end{proof}

\end{document}